\documentclass[a4paper,USenglish,cleveref,autoref,thm-restate,numberwithinsect,nameinlink,pagebackref]{lipics-v2021}
\nolinenumbers

\usepackage[utf8]{inputenc}
\usepackage[T1]{fontenc}
\usepackage{amsmath, amsthm, amssymb}
\usepackage{mathtools}
\usepackage{tabularx}
\usepackage{booktabs}
\usepackage{multirow}
\usepackage{dsfont}
\usepackage[]{todonotes} %
\usepackage[numbers,sort]{natbib}
\usepackage{nicefrac}

\usepackage[nolist, nohyperlinks]{acronym}
\usepackage{xspace}
\usepackage[dvipsnames]{xcolor}

\usepackage{tikz}
\usetikzlibrary{matrix,fit,decorations.pathreplacing,shapes,patterns,patterns.meta,intersections,positioning,chains,calc,backgrounds,arrows.meta,math,decorations.markings,backgrounds}

\tikzstyle{mycircle}=[circle,draw=black,fill=black!25,fill opacity = 0.3,text opacity=1,inner sep=0pt,minimum size=18pt,font=\small]
\tikzstyle{smallcircle}=[circle, draw=black, fill=white, inner sep=0pt, minimum size=15pt, font=\tiny]
\tikzstyle{vertex_small}=[circle,draw,fill, inner sep=2pt]
\tikzset{
	position/.style args={#1:#2 from #3}{
		at=(#3.#1), anchor=#1+180, shift=(#1:#2)
	}
}

\definecolor{italyGreen}{RGB}{0, 146, 70}
\definecolor{italyRed}{RGB}{206, 43, 55}

\Crefname{theorem}{Thm{.}}{Thms{.}}
\crefname{theorem}{Theorem}{Theorems}

\newcommand{\NN}{\mathds{N}} %
\newcommand{\QQ}{\mathds{Q}} %
\newcommand{\ZZ}{\mathds{Z}} %

\newcommand{\bigO}{\mathcal{O}} %

\DeclareMathOperator{\classNP}{NP}
\DeclareMathOperator{\tw}{tw}
\DeclareMathOperator{\cw}{cw}
\DeclareMathOperator{\FPT}{FPT}

\DeclareMathOperator{\clos}{clos}
\DeclareMathOperator{\W}{W}
\DeclareMathOperator{\vi}{vi}
\DeclareMathOperator{\ml}{ml}

\newtheorem{reduction}{Reduction Rule}

\newcommand{\taup}{\tau_{\rho}\xspace}

\newcommand{\bdvd}{\textsc{BDVD}\xspace}
\newcommand{\BDVD}{\textsc{Bounded Density Vertex Deletion}\xspace}
\newcommand{\bded}{\textsc{Bounded Density Edge Deletion}\xspace}

\newcommand{\vc}{\textsc{Vertex Cover}\xspace}
\newcommand{\fvs}{\textsc{Feedback Vertex Set}\xspace}
\newcommand{\ds}{\textsc{Dissociation Set}\xspace}
\newcommand{\rubp}{\textsc{Restricted Unary Bin Packing}\xspace}
\newcommand{\bdd}{\textsc{Bounded Degree Vertex Deletion}\xspace}
\newcommand{\coc}{\textsc{Component Order Connectivity}\xspace}

\acrodef{mwl}[M$W$-L]{\textsc{Max $W$-Light}}

\newcommand{\degphi}[0]{\deg_\phi^-}
\newcommand{\Degphi}[0]{\Delta_\phi^-}
\newcommand{\degfrac}[1]{\deg_{#1}^-}
\newcommand{\Degfrac}[1]{\Delta_{#1}^-}

\newcommand{\degexphi}{\degfrac{(\phi, w)}}
\newcommand{\Degexphi}{\Degfrac{(\phi, w)}}

\usepackage{mathtools}
\newcommand{\ceq}{\ensuremath{\coloneqq}}
\renewcommand{\epsilon}{\varepsilon}

\newcommand{\problemDef}[3]{%
\begin{center}
	\setlength{\tabcolsep}{2pt}
	\begin{tabular}{@{}lp{12cm}@{}}
		\multicolumn{2}{@{}l}{\textsc{#1}} \\%
		\textbf{Input:} & #2 \\%
		\textbf{Question:} & #3 \\%
	\end{tabular}
\end{center}%
}

\usepackage{etoolbox}

\newcommand{\appsymb}{$\star$}
\newcommand{\appref}[1]{{\hyperref[proof:#1]{\appsymb}}}

\newcommand{\appendixproof}[2]{%
  \gappto{\appendixProofs}
  {
    \subsection{Proof of \cref{#1}}\label{proof:#1}
    #2
  }
}

\makeatletter
\let\@mkboth\@gobbletwo
\makeatother
\title{On the Parameterized Complexity of Bounded-Density Vertex Deletion}

\author{Jakob Raupach}{Technische Universität Berlin, Germany}{raupach@campus.tu-berlin.de}{}{}

\author{Tom-Lukas Breitkopf}{Algorithmics and Computational Complexity, Technische Universität Berlin, Germany}{t.breitkopf@tu-berlin.de}{https://orcid.org/0009-0008-2875-1945}{}

\author{Anton Herrmann}{Algorithmics and Computational Complexity, Technische Universität Berlin, Germany}{a.herrmann@tu-berlin.de}{https://orcid.org/0009-0008-8473-9043}{}

\author{André Nichterlein}{Algorithmics and Computational Complexity, Technische Universität Berlin, Germany}{andre.nichterlein@tu-berlin.de}{https://orcid.org/0000-0001-7451-9401}{}

\authorrunning{J.\ Raupach, T.-L.\ Breitkopf, A.\ Herrmann, A.\ Nichterlein}

\Copyright{Jakob Raupach, Tom-Lukas Breitkopf, Anton Herrmann, André Nichterlein}



\hideLIPIcs %

\EventEditors{Michal Kouck\'{y} and Daniela Petri\c{s}an}
\EventNoEds{2}
\EventLongTitle{51st International Symposium on Mathematical Foundations of Computer Science (MFCS 2026)}
\EventShortTitle{MFCS 2026}
\EventAcronym{MFCS}
\EventYear{2026}
\EventDate{August 24--28, 2026}
\EventLocation{Paris, France}
\EventLogo{}
\SeriesVolume{386}
\ArticleNo{62}

\begin{document}
\maketitle

\begin{abstract}
We explore the parameterized complexity of \BDVD{} (\bdvd{}): given a graph~$G$, an integer budget~$k$, and a target density~$\taup$, the task is to determine whether the density (i.e.\ number of edges divided by number of vertices) of the densest subgraph of~$G$ can be reduced to at most~$\taup$ by deleting at most~$k$ vertices.
Our primary focus is on structural graph parameters related to treewidth, as the parameterized complexity of \bdvd with respect to treewidth was left as open question by Bazgan et al. [JCSS, 2025].
We resolve this question by showing $\W[1]$-hardness with respect to various parameters, including treedepth and feedback vertex number.
These results imply $\W[1]$-hardness with respect to treewidth.
We obtain positive results for parameters larger than treedepth and feedback vertex number, namely we show \bdvd is in~$\FPT$ parameterized by the max leaf number or vertex integrity.
Under the assumption that the target density~$\taup$ is a fixed constant the parameterized complexity landscape of \bdvd changes drastically, allowing a fixed-parameter tractable algorithm even for parameters smaller than treewidth, namely cliquewidth.
Altogether, our results provide a refined complexity landscape for \BDVD{}, sharply distinguishing between tractable and intractable parameter regimes under structural parameterizations.
\end{abstract}

\section{Introduction}\label{chap:intro}
The problem of identifying a densest subgraph constitutes a central topic in graph theory, with numerous applications across domains such as social network analysis, bioinformatics, and financial modeling~\cite{LANC24}. 
This problem has been extensively studied, with the first polynomial-time algorithms proposed more than four decades ago~\cite{PQ82}. 
In this work, we examine the stability of densest subgraphs under vertex deletions in the input graph.

We investigate \BDVD{} (\bdvd), which asks whether it is possible to delete at most~$k$ vertices so that the density of the densest subgraph falls to at most a specified threshold~$\taup$ (see \cref{fig:bdvd_example} for an example).
\begin{figure}[t]
	\centering
	\begin{tikzpicture}[yscale=0.9,dot/.style={draw,circle,inner sep=4pt},tdot/.style={draw,circle,inner sep=4pt,line width=0.3mm, red}]
		\node[dot] at (0.5, -1) (v1) {};
		\node[dot] at (2, 0) (v2) {};
		\node[dot] at (3, -1) (v3) {};
		\node[dot] at (3, 0) (v4) {};
		\node[dot, fill] at (4, 0) (v5) {};
		\node[dot] at (5, 0) (v6) {};
		\node[dot] at (5, -1) (v7) {};
		\node[dot] at (6, -1) (v8) {};
		\node[dot, fill] at (6, 0) (v9) {};
		\node[dot] at (7, 0) (v10) {};
		\node[dot] at (7, -1) (v11) {};
		\node[dot] at (8, -1) (v12) {};

		\draw
			(v1) -- (v2)
			(v2) -- (v4)
			(v5) -- (v4)
			(v3) -- (v5)
			(v3) -- (v2)
			(v3) -- (v4);
		\draw
			(v5) -- (v6)
			(v6) -- (v7)
			(v6) -- (v9)
			(v7) -- (v8)
			(v8) -- (v9)
			(v9) -- (v10)
			(v11) -- (v9)
			(v12) -- (v11)
			(v12) -- (v10);
		
		\node[inner sep = 4pt, very thick, draw,rounded corners, dashed, fit=(v2) (v3) (v5)] {};
		\node[inner sep = 12pt, very thick, draw,rounded corners, dashed, fit=(v2) (v3) (v9)] {};
		\node[inner sep = 8pt, very thick, draw,rounded corners, dashed, fit=(v6) (v7) (v12)] {};

	\end{tikzpicture}
	\caption{
	Example graph for \bdvd{} with~$k = 2$ and~$\taup = 1$.
	The remaining graph may only consist of connected components, each containing at most one cycle.
	Several subgraphs whose density exceeds $\taup$ are indicated,
	together with a possible solution (shown as filled vertices). Note that
	at least one vertex of each indicated subgraph must be deleted in a solution.
	\label{fig:bdvd_example}
	}
\end{figure}
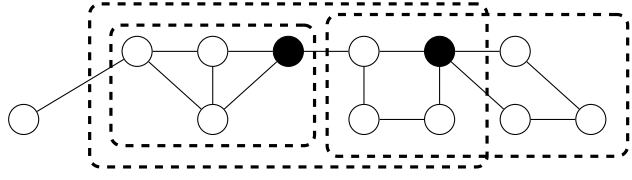
Here, the density of a graph is defined as the ratio of the number of edges to the number of vertices. 
This problem was first introduced by Bazgan et al. \cite{BNV25}
and observed to generalize several classical vertex deletion problems:
namely, \vc{} for $\taup = 0$, \ds{} for $\taup = \nicefrac{1}{2}$, and \fvs{} for $\taup = 1-\nicefrac{1}{n}$ (where $n$ is the number of vertices in the graph). 
Moreover, for integral target densities, \bdvd{} actually coincides with $d$-\textsc{Orientable Deletion}~\cite{HKLOS20}; this follows from a characterization of graph density by Charikar~\cite{CHAR00}.

Densest subgraph problems have been studied extensively since the 1970s and have experienced renewed attention in recent years~\cite{LANC24}. 
Early algorithmic approaches were based on flow techniques~\cite{LANC24,PQ82,Gol84}.
The aforementioned characterization for graph density by Charikar~\cite{CHAR00} is based on the dual of a linear programming approach for computing the density of a densest subgraph.
This led to the notion of fractional orientations, which have since become a common tool in the study of dense subgraphs~\cite{NICH26,CHAN25}.

An \emph{orientation} of an undirected graph~$G$ assigns a direction to each edge.
A graph is called \emph{$d$-orientable}, for~$d \in \NN$, if there is an orientation resulting in a digraph with maximum indegree at most~$d$\footnote{In parts of the literature~\cite{AsahiroJMO15, BodlaenderOO18}, the bound~$d$ is for the maximum outdegree. To keep consistency with the characterization by Charikar~\cite{CHAR00}, we bound the indegree, which is completely symmetrical~\cite{HKLOS20}.}.
The task in $d$-\textsc{Orientable Deletion} is to remove a minimum number of vertices so that the resulting graph becomes $d$-orientable. 
A \emph{fractional orientation} assigns parts of an edge to each endpoint. 
More precisely, a fractional orientation is a function~$\phi$ assigning two non-negative rationals~$\phi(e)^u$ and~$\phi(e)^v$ with $\phi(e)^u + \phi(e)^v = 1$ to each edge $e \ceq \{u, v\} \in E(G)$.
The characterization of Charikar~\cite{CHAR00} states that a densest subgraph has density~$\taup$ if and only if the graph admits a fractional orientation with maximum indegree~$\taup$.
Moreover, as observed by Bentert et al.~\cite{NICH26} and Chandrasekaran et al.~\cite{CHAN25}, if~$\taup \in \NN$, then there also exists an integral orientation\footnote{To avoid ambiguity between \emph{orientation} and \emph{fractional orientation}, we call the former \emph{integral orientation}.} with maximum indegree~$\taup$ (see \cref{sec:cliquewidth} for more details).
Hence, $d$-\textsc{Orientable Deletion} is equivalent to \bdvd{} with~$\taup = d$.

Computing an integral or fractional orientation minimizing the indegree can be done in polynomial-time via a simple reduction to flow~\cite{NICH26}. 
Introduced by Asahiro et al.~\cite{AsahiroJMO15}, $d$-\textsc{Orientable Deletion} turns out to be a computationally challenging problem:
It is W[2]-hard with respect to the number of vertices to remove, W[1]-hard with respect to the cliquewidth of the input graph and not approximable within a better factor than~$\ln n$~\cite{HKLOS20}.
Algorithmic results include fixed-parameter tractability for treewidth and for the combined parameter cliquewidth plus~$d$~\cite{BodlaenderOO18, HKLOS20}.
Asahiro et al.~\cite{AJMO16} provided a $O(\log n)$-approximation for $d$-\textsc{Orientable Deletion}.
Chandrasekaran et al.~\cite{CHAN25} gave the same approximation ratio for problems generalizing \bdvd{}.
All hardness results directly carry over to \bdvd{}.

The classical computational complexity of \bdvd{} was studied by Bazgan et al.~\cite{BNV25}.
They showed that the problem is solvable in polynomial time on trees, but is $\classNP$-hard on planar bipartite graphs of maximum degree~4. 
Unaware of the connection to $d$-\textsc{Orientable Deletion}, they reprove the $\W[2]$-hardness of \bdvd{} with respect to the number of vertices to be deleted.
They also show fixed-parameter tractability when parameterized by the vertex cover number.
In the same work, Bazgan et al. considered the sister problem \bded{} (deleting edges instead of vertices; thereby generalizing the polynomial-time solvable \textsc{Maximum Cardinality Matching} for~$\taup = 1/2$ and \textsc{Feedback Edge Set} for~$\taup = 1-1/n$) and proved that it is $\W[1]$-hard when parameterized by the solution size combined with the feedback edge number and pathwidth.
Since the feedback edge number upper-bounds the treewidth, this result also implies $\W[1]$-hardness with respect to the combined parameter treewidth and solution size.
Whether an analogous hardness result holds for \bdvd{} remained open.

Heavily relying on fractional orientations, Bentert et al.~\cite{NICH26} showed that \bded{} is polynomial-time solvable whenever the target density is a multiple of~$1/2$ or strictly less than~$2/3$, while it is $\classNP$-hard for all other values.
They also proved that \bded{} is fixed-parameter tractable for the parameter treewidth when the target density is fixed via a reduction to \textsc{General Factors}.

\subparagraph*{Our results.}
All the listed special cases of \bdvd{} mentioned above are fixed-parameter tractable with respect to the treewidth; whether this is also true for \bdvd{} in general was posed as open question by Bazgan et al.~\cite{BNV25}.
As can already be seen at the special case $d$-\textsc{Orientable Deletion}, \bdvd{} has some aspects of a number problem: 
After deleting vertices, the edges have to be distributed in a balanced way to their endpoints.
If~$d$ becomes large, then this number problem aspect becomes quite challenging: 
The fixed-parameter tractability of $d$-\textsc{Orientable Deletion} with respect to cliquewidth and~$d$ is based on a dynamic program which stores partial solutions for each maximum indegree~$d' \le d$~\cite{HKLOS20}.
The W[1]-hardness with respect to cliquewidth alone~\cite{HKLOS20} shows that this approach of dealing with the inherent number problem can presumably not be improved.
In \cref{sec:cliquewidth}, we extend this dynamic program to \bdvd{} for target density~$\taup = \nicefrac{p}{q}$ with~$p,q \in \NN$, resulting in a running time of~$2^{O(p^7\log q + \cw \cdot p \log p)} n^{O(1)}$, where~$\cw$ is the cliquewidth of the input graph.

While structural observations (any graph of treewidth~$\tw$ is easily $\tw$-orientable) allows to deduce fixed-parameter tractability for $d$-\textsc{Orientable Deletion}, this turns out to be insufficient for \bdvd{}: We provide a dynamic program running in time $(\tw \cdot q)^{O(\tw^2)} \cdot n^{O(1)}$, where~$\tw$ is the treewidth of the input graph.
In \cref{sec:hardness}, we establish that the dependency on~$q$ cannot be avoided: we show that \bdvd is $\W[1]$-hard when parameterized by treewidth, and larger parameters like treedepth and feedback vertex set.
For unrestricted target densities, we also prove in \cref{sec:algos} fixed-parameter tractability with respect to the even larger parameters vertex integrity and maximum leaf number, respectively.
An overview of our results is shown in \cref{fig:results}.
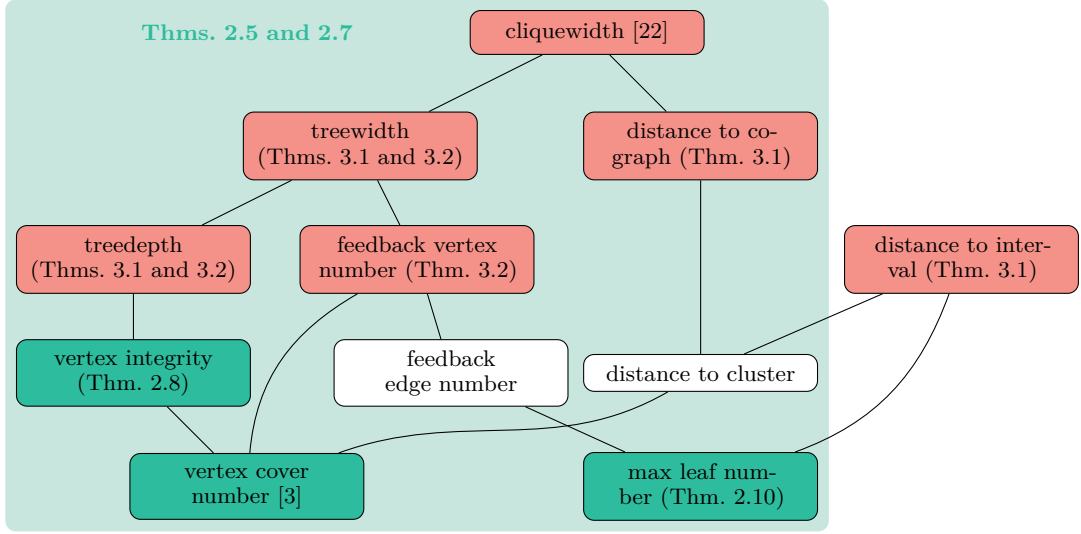
\begin{figure}
 \centering
 \begin{tikzpicture}[
  yscale=1.5,
  xscale=1.5,
  every node/.style={draw,rectangle,rounded corners,font=\footnotesize,align=center,text width=8em,inner sep=4pt},
  fpt/.style={fill=SeaGreen},
  hard/.style={fill=Salmon},
  unknown/.style={fill=white}]
	\def\horizontalspace{2}
	\def\verticalspace{2}
	\node[unknown] (fes) at (2.8,1) {feedback edge number};
	\node[fpt] (vc) at (1,0) {vertex cover number~\cite{BNV25}};
	\node[hard] (fvs) at (2.5,2) {feedback vertex number (\Cref{thm:whard_td_fvn})};
	\node[fpt] (vi) at (0,1) {vertex integrity (\Cref{thm:vertex_integrity})};
	\node[fpt] (ml) at (5,0) {max leaf number (\Cref{thm:max_leaf_number})};
	\node[unknown] (cluster) at (5,1) {distance to cluster};
	\node[hard] (td) at (0,2) {treedepth (\Cref{thm:whard_td_co_interval,thm:whard_td_fvn})};
	\node[hard] (interval) at (7.3,2) {distance to interval (\Cref{thm:whard_td_co_interval})};
	\node[hard] (tw) at (2,3) {treewidth (\Cref{thm:whard_td_co_interval,thm:whard_td_fvn})};
	\node[hard] (cograph) at (5,3) {distance to cograph (\Cref{thm:whard_td_co_interval})};
	\node[hard] (cw) at (4,4) {cliquewidth~\cite{HKLOS20}};
	\node[draw=none,fill=none] (fix) at (1,4) {\textcolor{SeaGreen}{\textbf{\Cref{thm:fpt_cw_fixed_density,thm:fpt_tw_fixed_density}}}};

	\draw (ml) -- (fes) -- (fvs);
	\draw (vc) to[out=85,in=210] (fvs);
	\draw (vc) -- (vi);
	\draw (vc) to[out=20,in=210] (cluster);
	\draw (vi) -- (td);
	\draw (cluster) -- (interval);
	\draw (cluster) -- (cograph);
	\draw (td) -- (tw);
	\draw (tw) -- (cw);
	\draw (cograph) -- (cw);
	\draw (fvs) -- (tw);
	\draw (ml) to[out=20,in=250] (interval);

	\begin{scope}[on background layer]
	\node[draw=none,fill=SeaGreen!30, rounded corners, fit = (vc) (vi) (cw) (ml)]{};
	\end{scope}
 \end{tikzpicture}
  \caption{Hasse diagram of parameterized complexity of \bdvd with parameters growing from top to bottom.
  Green indicates fixed-parameter tractability while red indicates~$\W[1]$-hardness.
  Hardness even holds for the combined parameters treedepth + feedback vertex number and treedepth + distance to cograph + distance to interval graph.
  Indicated with a light green background are all the parameters for which \bdvd becomes fixed-parameter tractable if the target density~$\taup$ is constant.}
  \label{fig:results}
\end{figure}

\subparagraph*{Graph theory \& problem definition.}
For a graph $G$, we write $V(G)$ and $E(G)$ to denote the set of vertices and edges of $G$, respectively.
For an edge~$\{u,v\} \in E(G)$ we sometimes use~$uv$ as a shorthand.
We say $H$ is a \emph{subgraph} of~$G$, denoted~$H \subseteq G$, if~$H$ is a graph with $V(H) \subseteq V(G)$ and $E(H) \subseteq E(G)$.
Further for any set of vertices~$V' \subseteq V(G)$ we denote by $G[V']$ the \emph{induced} subgraph of $G$ on $V'$.
For any set of vertices~$V' \subseteq V(G)$ we denote by~$G - V' \ceq G[V(G) \setminus V']$ the graph obtained from $G$ by deleting the vertices~$V'$ and all incident edges. 
If $V'$ consists only of one vertex $v$, then we also write $G-v$ instead of $G-\{v\}$.
We denote the degree of a vertex $v \in V(G)$ by~${\deg_G(v) = |N_G(v)|}$.
If the graph is clear from the context, then we omit $G$.
A path~$P$ in~$G$ is a sequence~$(v_1, \dots, v_n)$ of pairwise different vertices such that~$v_{i}v_{i+1} \in E(G)$ for all~$i \in [n-1]:= \{1, \dots, n-1\}$.
The vertices~$v_1$ and~$v_n$ are called the endpoints of~$P$.
The complete graph (or \emph{clique}) on~$t \in \NN$ vertices is denoted by~$K_t$, and the complete bipartite graph (or \emph{biclique}) with parts of sizes~$s,t \in \NN$ is denoted by~$K_{s,t}$.
We define the \emph{density} $\rho(G)$ of a non-empty graph~$G$ with $\rho(G) \ceq |E(G)|/|V(G)|$ and the density of the empty graph as 0.
We also define the \emph{density of the densest subgraph} of a graph~$G$ as~$\rho^*(G) \coloneqq \max_{H \subseteq G} \rho(H)$.
A subgraph $H \subseteq G$ attaining this maximum is called a \emph{densest subgraph} of~$G$.

The problem definition of \bdvd{} as introduced by Bazgan et al.~\cite{BNV25} is as follows:
\problemDef{\BDVD}
{
	A graph $G$, an integer $k \in \NN$ and a rational number $\taup \ge 0$.
}
{
	Is there a subset $S \subseteq V(G)$ with $|S| \le k$ such that $\rho^* (G-S) \le \taup$?
}
\subparagraph*{Parameterized complexity.}
Parameterized complexity provides a multivariate approach for measuring the time complexity of a computational problem~\cite{PARAMALG,DF13}.
An instance~$(x,k)$ of a parameterized problem consists of a classical instance~$x$ together with a number~$k$, called the parameter.
A parameterized problem is \emph{fixed-parameter tractable} (FPT) if it can be decided in time~$f(k) \cdot |x|^{O(1)}$ for some computable function~$f$.
If for some~$t \ge 1$ a parameterized problem is~$\W[t]$-hard with respect to parameter~$k$, then it is considered unlikely to be fixed-parameter tractable when parameterized by~$k$.

\section{Algorithmic Results}\label{sec:algos}
We first consider the case in which~$p,q \in \NN$ are part of the parameter indicating the ``encoding size'' of~$\taup = \nicefrac{p}{q}$.
We show fixed-parameter tractability with respect to the combined parameter cliquewidth + $p + q$ by means of dynamic programming.
As indicated in \cref{fig:results} this immediately implies fixed-parameter tractability for most other parameters considered in this work when combined with~$p$ and~$q$.
As the larger parameter treewidth bounds the cliquewidth exponentially, however, the running time of the dynamic program becomes double-exponential in the parameter.
We therefore provide another dynamic program for the parameter treewidth combined with~$p$ and~$q$ with a single-exponential running time.

As will be established in \cref{sec:hardness}, \bdvd becomes much harder if $p+q$ is not part of the parameter.
For this setting we therefore turn to larger parameters.
Specifically, we show fixed-parameter tractability for the parameters vertex integrity, thereby strengthening a previous result for vertex cover number~\cite{BNV25}, and max leaf number.
As can be seen in \cref{fig:results}, this draws a somewhat close line to the border of tractability established later.
Both algorithmic approaches use some initial branching to simplify the instance and then use Integer Linear Programming to tackle the number problems that remain.

\subsection{Cliquewidth and Target Density}
\label{sec:cliquewidth}
Here, we show that \BDVD is fixed-parameter tractable parameterized by cliquewidth, $p$, and~$q$ where~$p,q \in \NN$ with~$\taup = \nicefrac{p}{q}$.
Recall that the cliquewidth~$\cw$ of a graph~$G$ is the smallest number of labels such that~$G$ can inductively be obtained as a labeled graph from the following four operations \cite{CourcelleMR00}.
\begin{enumerate}[i)]
 \item Introduce~$\circ_i$: The graph consisting of a single vertex with label~$i$.
 \item Union~$G_1 \oplus G_2$: Creates the disjoint union of the labeled graphs~$G_1$ and~$G_2$.
 \item Relabel~$\rho(i_1,i_2)$: Relabels all vertices of label~$i_1$ with the label~$i_2$.
 \item Join~$\eta(i_1,i_2)$: All edges between the vertices of the labels~$i_1$ and~$i_2$ are added.
\end{enumerate}
An \emph{$\ell$-expression}~$T$ is a rooted binary tree in which every node is associated with one of the four operations, $\ell$ labels are used and the following properties hold.
\begin{itemize}
 \item Each leaf is associated with an introduce operation.
 \item Each node with two children is associated with an union operation.
 \item Each node with exactly one child is associated with a relabel or join operation.
\end{itemize}
If~$G$ is obtained as a labeled graph at the root, then we call~$T$ an \emph{$\ell$-expression} of~$G$.
The subgraph of~$G$ constructed up to node~$t$ of~$T$ is denoted by~$G_t$.
Given a graph~$G$ of cliquewidth~$\cw$, one can compute a~$(2^{\cw + 1} - 1)$-expression of~$G$ in FPT time with respect to~$\cw$ \cite{Korhonen024}.

In the following it is helpful to think about fractional orientations instead of subgraph densities.
The question then is whether we can delete~$k$ vertices from~$G$ such that there exists a fractional orientation with maximum indegree at most~$\nicefrac{p}{q}$?
Formally a fractional orientation is a function~$\phi$ assigning two non-negative rationals~$\phi(e)^u$
and~$\phi(e)^v$ with $\phi(e)^u + \phi(e)^v = 1$ to each edge
$e \ceq \{u, v\} \in E(G)$ of graph~$G$.
We denote with
\[
	\degphi(v) \ceq \sum_{u \in N(v)} \phi(\{u, v\})^v \qquad \text{and} \qquad \Degphi \ceq \max_{v \in V(G)} \degphi(v)
\]
the \emph{indegree} of a vertex $v \in V(G)$ and the \emph{maximum indegree} of a graph~$G$ with respect to~$\phi$.
The connection between fractional orientations and the density of the densest subgraph is formalized by the following lemma.
\begin{lemma}[Charikar \cite{CHAR00}]
	For every graph $G$, there exists a fractional orientation~$\phi$ such that~$\rho^*(G) = \Degphi(G)$.
	\label{lem:charikar}
\end{lemma}

In this context, a fixed target density allows us to restrict ourselves to the existence of~\emph{$q$-quantized fractional orientations} which we define next.
\begin{definition}
    Let $G$ be a graph and $q \in \NN$. A fractional orientation $\phi$ is
    called \emph{$q$-quantized} if for every edge $e = \{u, v\} \in E(G)$
    the following holds: $\phi(e)^u = \nicefrac{p_u}{q}$ and~$\phi(e)^v = \nicefrac{p_v}{q}$ for $p_u, p_v \in \NN$.
\end{definition}

\begin{observation}
	Let $G$ be a graph and $\phi$ a $q$-quantized fractional orientation.
	Then, for all $v \in V(G)$ we have~$\degphi(v) = \nicefrac{p_v}{q}$ for some $p_v \in \NN$.
\end{observation}

A consequence of the next lemma is that it is sufficient to consider~$q$-quantized fractional orientations for the remainder of the section.
In essence the proof takes the flow construction of Bentert et al.~\cite{NICH26} which is depicted in \cref{fig:q_quantized_flow} and uses it to compute a fractional orientation.
By scaling the network and using known integral properties for maximum flows the claim follows.

\begin{lemma}
\label{lem:q_quantized}
	Let $G$ be a graph. For every $\taup = \nicefrac{p}{q} \ge \rho^*(G)$,
	there exists a $q$-quantized fractional orientation $\phi$ with
	$\Degphi(G) \le \taup$.
\end{lemma}
\begin{proof}
	We prove the statement by scaling the flow network~$N$ by Bentert et al.~\cite{NICH26} and then using known integral properties for maximum flows.
	The vertex set of the flow network~$N$ contains a source~$s$, a sink~$t$, all vertices from~$G$ and for every edge~$e \in E(G)$ one vertex~$v_e$.
	Formally this means~$V(N) := \{s,t\} \cup V(G) \cup \{v_e \mid e \in E(G)\}$.
	For every~$e \in E(G)$ add the arc~$(s,v_e)$ with capacity one and for every~$u \in V(G)$ add the arc~$(u,t)$ with capacity~$\taup = \nicefrac{p}{q}$ to~$N$.
	Moreover, for every edge~$e = uw \in E(G)$ introduce two arcs~$(v_e,u)$ and~$(v_e,w)$ with capacity one.
	A sketch of the constructed network can be seen in~\cref{fig:q_quantized_flow}.

	By~\cref{lem:charikar} we know that there exists a fractional orientation~$\phi_1$ in the graph~$G$ with~$\Delta_{\phi_1}^- = \rho^*(G) \leq \taup$.
	Consequently, the following flow~$f_1$ is feasible in~$N$:
	\begin{equation*}
	f_1(a) \coloneqq
		\begin{cases}
			1 & \text{if } a = (s,v_e) \text{ for some } e \in E(G) \\
			\phi_1(e)^u & \text{if } a = (v_e,u) \\
			\deg_{\phi_1}^-(u) & \text{if } a = (u,t) \text{ for some } u\in V(G).
		\end{cases}
	\end{equation*}
	Note that~$f_1$ is a maximum flow for~$N$ since all outgoing arcs from the source have a maximal throughput of~$1$.
	Now consider the network~$qN$, which is obtained from~$N$ by multiplying all capacities with~$q$.
	It is clear that~$qf_1$, that is the flow~$f_1$ scaled by the factor~$q$, is a maximum flow in~$qN$.
	Since all capacities in~$qN$ are integral, we know that there exists an integral maximum flow.
	Let~$qf_2$ be this integral maximum flow.
	Then~$f_2$ is a maximum flow in~$N$ in which the flow on every arc is a multiple of~$\nicefrac{1}{q}$ since~$qf_2$ is integral.
	This allows us to define the $q$-quantized fractional orientation~$\phi_2$ by setting~$\phi_2(e)^u := f_2((v_e,u))$.
	Recall that the arc~$(v,t)$ has capacity~$\taup$ in~$N$ for every~$v \in V(G)$ and therefore by construction~$\Delta_{\phi_2}^- \leq \taup$.
\end{proof}

	\begin{figure}
	\centering
	\begin{tikzpicture}[myarrow/.style={-Stealth},node distance=1cm and 1cm]
			\begin{scope}
				\begin{scope}[node distance=1cm and 2cm]
					\node[mycircle] (v1) {$v_1$};
					\node[mycircle, right=of v1] (v2) {$v_2$};
					\node[mycircle, below=of v1] (v3) {$v_3$};
					\node[mycircle, below=of v2] (v4) {$v_4$};
				\end{scope}

				\foreach \i/\j/\txt/\p in {%
					v1/v2/$e_1$/above,
					v2/v3/$e_2$/above,
					v3/v4/$e_3$/above,
					v1/v3/$e_4$/left,
					v2/v4/$e_5$/right}
					\draw [-] (\i) -- node[font=\small,\p] {\txt} (\j);
			\end{scope}

			\begin{scope}[xshift=8cm,yshift=1.4cm]
				\begin{scope}[yscale=1.4,xscale=1.6]
					\node[mycircle] (s) at (0,0) {$s$};
					\foreach \i in {1,...,5}
					{
						\node[mycircle] (e\i) at (-3 + \i, -1) {$v_{e_\i}$};
					}
					\foreach \i in {1,...,4}
					{
						\node[mycircle] (v\i) at (-2.5 + \i, -2) {$v_\i$};
					}
					\node[mycircle] (t) at (0,-3) {$t$};
				\end{scope}

				\foreach \i/\j in {s/e2, s/e3, s/e4, s/e5, e1/v2, e2/v2, e2/v3, e3/v3, e3/v4, e4/v1, e4/v3, e5/v2, e5/v4, v2/t, v3/t, v4/t}
					\draw [myarrow] (\i) -- (\j);

				\foreach \i/\j/\txt/\p in {s/e1/$1$/left, e1/v1/$1$/left, v1/t/$\taup = \frac{p}{q}$/left}
					\draw [myarrow] (\i) -- node[font=\small,\p, xshift=-10pt] {\txt} (\j);

			\end{scope}
		\end{tikzpicture}
	\caption{A sketch of the flow network which is constructed and scaled in~\cref{lem:q_quantized}.
	We remark that the figure is taken from Bentert et al.~\cite{NICH26}.}
	\label{fig:q_quantized_flow}
\end{figure}
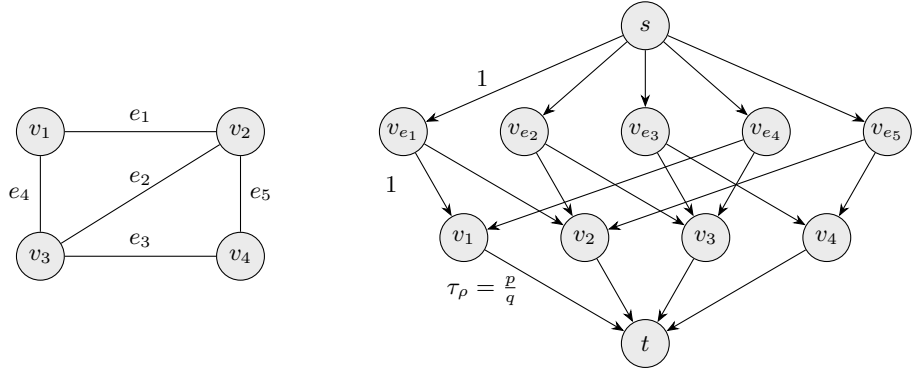

The idea of our algorithm is to extend the dynamic programming algorithm from Hanaka
et al.~\cite{HKLOS20} which is an extension of an algorithm from Bodlaender et al.~\cite{BodlaenderOO18}.
The main idea is bottom-up dynamic programming over a given cliquewidth expression~$T$ to keep track of all possible indegree signatures.
For a node~$t$ of~$T$, a pair~$(k,A)$, where~$k$ is a natural number and~$A = (A_{i,j})_{j\in [1,\cw],j \in [0,p]}$ is a~$\cw \times (p+1)$ matrix, is an \emph{indegree signature} if there exists~$K_t \subseteq V(G_t)$ with~$|K_t| \leq k$ and a fractional orientation~$\phi_t$ of~$G_t - K_t$ such that the number of label~$i$ vertices with indegree~$\nicefrac{j}{q}$ in~$G_t - K_t$ is exactly~$A_{i,j}$.
For a label~$i$, we refer to~$j \in [0,p]$ as an~\emph{indegree class} of~$i$.
The improvement from~\cite{HKLOS20} over~\cite{BodlaenderOO18} is the combinatorial observation that the exact size of an indegree class does not matter if it is large enough ($> 3p^3$ for us) as in this case all newly created edges can be oriented towards the vertices from the large indegree class.
To further extend the algorithm from Hanaka et al.~\cite{HKLOS20}, we use the fact that we can restrict ourselves to~$q$-quantized fractional orientations when the maximum permitted indegree is~$\nicefrac{p}{q}$.

\begin{theorem}
	\label{thm:fpt_cw_fixed_density}
	Given a~$\cw$-expression of the graph, \BDVD can be solved in~$q^{O(p^7)} p^{O(\cw \cdot p)} \cdot n^{O(1)}$.
\end{theorem}

\begin{proof}
	Let~$G$ be the input graph and~$T$ a given~$\cw$-expression of~$G$.
	In the following we state how to compute all indegree signatures~$(k,A)$ for the nodes of~$T$ in a bottom-up manner.
	The smallest number~$k$ for which there is an indegree signature~$(k,A)$ of the root~$r$ equals the minimum number of vertices which have to be deleted such that there exists a fractional orientation with maximum indegree at most~$\nicefrac{p}{q}$.
	Without explicitly mentioning it again, we always assume that we set an entry of the matrix~$A$ of an indegree signature~$(k,A)$ to 'large' if the respective value would be larger than~$3p^3$.
	We need to distinguish between four cases for a node~$t$ of~$T$, where the first three cases are essentially the same as in~\cite{BodlaenderOO18, HKLOS20}.

	\emph{Introduce Node~$\circ_i$:} We have two indegree signatures.
	If~$k=0$, then~$A_{i,0} = 1$ and all other entries of~$A$ are zero.
	If~$k=1$, then~$A$ is the zero matrix.

	\emph{Union Node~$G_{t'} \oplus G_{t''}$:} Let~$t'$ and~$t''$ be the children of~$t$.
	Two indegree signatures~$(k',A')$ and~$(k'',A'')$ of the nodes~$t'$ and~$t''$ respectively are merged to~$(k:=k'+k'',A)$ where we set~$A_{i,j} \coloneqq A'_{i,j} + A''_{i,j}$ for all~$j$.

	\emph{Relabel Node~$\rho(i_1,i_2)$:} Let~$t'$ be the child of~$t$.
	For every indegree signature~$(k',A')$ of~$t'$ we obtain an indegree signature~$(k:= k',A)$ of~$t$ by setting~$A_{i_1,j} \coloneqq 0, A_{i_2,j} \coloneqq A'_{i_1,j} + A'_{i_2,j}$ and~$A_{i,j} \coloneqq A'_{i,j}$ fo all~$i\notin \{i_1,i_2\}$ and all~$j$.

	\emph{Join Node~$\eta(i_1,i_2)$:} Let~$t'$ be the child of~$t$.
	Join nodes are the only nodes where edges are created.
	As we have multiple options how to orient the new edges, we potentially compute a lot of indegree signatures~$(k,A)$ of~$t$ from one indegree signature~$(k',A')$ of~$t'$.
	If both rows~$A'_{i_1}$ and~$A'_{i_2}$ contain an entry with the label 'large', then the indegree signature~$(k',A')$ of~$t'$ can be discarded since~$\rho(K_{2p+1,2p}) > \nicefrac{p}{q}$.
	Now suppose the row~$A'_{i_2}$ contains no label 'large' and let~$A'_{i_1,j_1}$ be an entry of the row~$A'_{i_1}$.
	There are two cases.
	\begin{itemize}
	 \item $A'_{i_1,j_1}$ contains the label 'large'.
	 In this case we can orient all newly created edges towards the vertices from~$A'_{i_1,j_1}$ and move them to the larger indegree class~$j_1 + sq$ where~$s := \sum_{j=0}^{p} A'_{i_2,j}$.
	 The following claim (adjusted from~\cite{HKLOS20} to our setting) shows that this does not change the smallest~$k$ for which there is an indegree signature of the root~$r$.
	 \begin{claim}
	  If there is an indegree signature~$(k,A_1)$ of the root~$r$ when not all of the created edges are fully oriented towards~$A'_{i_1,j_1}$, then there is also an indegree signature~$(k,A_2)$ of the root~$r$ when all of them are fully oriented towards~$A'_{i_1,j_1}$.
	 \end{claim}
	 \begin{claimproof}
		We show that there exists at least one vertex from the indegree class~$j_1$ that receives the full orientation of every edge introduced in ascendants of the node~$t'$ in~$T$ without exceeding the indegree~$\nicefrac{p}{q}$.
		This implies that all vertices from the same indegree class~$j_1$ of label~$i_1$ can receive the full orientation of these edges without exceeding the indegree~$\nicefrac{p}{q}$.
		Since~$A'_{i_1,j_1} > 3p^3$ and~$\rho(K_{2p+1,2p}) > \nicefrac{p}{q}$ it is clear that~$\sum_{j=0}^{p} A'_{i_2,j} < 2p$.
		As we are only considering~$q$-quantized fractional orientations, we know that every vertex with label~$i_2$ can receive orientation from at most~$p$ edges.
		Thus, more than~$3p^3 - 2p^2$ of the vertices from the indegree-class~$j_1$ receive the full orientation from all new edges to which they are incident and thus are moved to a higher indegree-class~$j' > j_1$.
		There are only~$p+1$ indegree-classes, so a vertex can move at most~$p$ times to a higher indegree-class.
		Consequently, more than $3p^3 - 2p^3$ vertices of~$j_1$ are assigned all newly created edges in the ascendants of~$t'$.
	 \end{claimproof}
	 \item $A'_{i_1,j_1}$ does not contain the label 'large'.
	 In this case, the algorithm needs to consider all possible~$q$-quantized fractional orientations of the new edges that are incident to the vertices from~$A'_{i_1,j_1}$.
	 Since~$A'_{i_1,j_1} \leq 3p^3$ and each of the~$p+1$ indegree classes of the label~$i_2$ also has size at most~$3p^3$, there are~$O(p^7)$ new edges between the indegree class~$j_1$ of the label~$i_1$ and vertices of the label~$i_2$.
	 Each edge has~$q+1$ possible orientations, and therefore the algorithm has to consider~$(q+1)^{O(p^7)}$ possible~$q$-quantized fractional orientations for~$A'_{i_1,j_1}$.
	 If one vertex would get indegree larger than~$\nicefrac{p}{q}$, then the respective fractional orientation is ignored.
	 Otherwise, the vertices are moved to the corresponding indegree classes in~$A$ according to the fractional orientation.
	\end{itemize}
	For each of the~$p+1$ indegree classes of the label~$i_1$, we create at most~$(q+1)^{O(p^7)}$ indegree signatures~$(k:=k',A)$ of~$t$ for a given indegree signature~$(k',A')$ of~$t'$.

	For an indegree signature~$(k,A)$ every entry of~$A$ contains either the label 'large' or a number from the interval~$[0,3p^3]$.
	Hence, for every node~$t$ of the cliquewidth expression there exist at most~$n \cdot (3p^3 + 2)^{\cw \cdot (p+1)}$ indegree signatures.
	Combined with the running time of the most expensive operation (join), this yields a running time of~$q^{O(p^7)} p^{O(\cw \cdot p)} \cdot n^{O(1)}$.
\end{proof}

\subsection{Treewidth and Target Density}
\cref{thm:fpt_cw_fixed_density} implies that \bdvd is in FPT parameterized by treewidth,~$p$ and~$q$ where~$p,q\in\NN$ with~$\taup=\nicefrac{p}{q}$ as the treewidth of a graph bounds in the cliquewidth~\cite{CO00,CR05}.
As the cliquewidth can be exponential in the treewidth, however, this gives a double-exponential running time.
Using dynamic programming over nice tree decompositions we obtain a single-exponential running time for the parameter treewidth and~$q$.
The algorithm follows the standard bottom-up dynamic programming approach on nice tree decompositions.
To keep the number of states manageable, we again use quantized fractional orientations, which restrict the number of possible orientations that need to be considered for each bag of the decomposition.
For every bag, the algorithm considers all subsets of the bag and evaluates all feasible \emph{weighted} fractional orientations for them, where the weights are used to track the indegree contributions of
already forgotten vertices adjacent to vertices contained in the subset, while simultaneously storing the minimum number of vertices that must be deleted so that the remaining graph admits a fractional orientation consistent with the considered weighted fractional orientation and with maximum indegree at most~$\taup$.
We can further show that~$p \le \tw \cdot q$ (otherwise no vertex needs to be deleted), which leads to the next theorem.
Proofs of statements marked with~\appsymb{} are deferred to the appendix.

\begin{restatable}[\appref{thm:fpt_tw_fixed_density}]{theorem}{fpttwfixeddensity}\label{thm:fpt_tw_fixed_density}
	Given an $n$-vertex graph $G$ together with a nice tree decomposition of
	width at most $\tw$ and a target density $\taup = \nicefrac{p}{q}$,
	then \BDVD{} can be solved in time $(q\cdot p)^{\bigO(\tw^2)} \cdot n^{\bigO(1)}$ and in time~$(\tw \cdot q)^{\bigO(\tw^2)} \cdot n^{\bigO(1)}$.
\end{restatable}

\appendixproof{thm:fpt_tw_fixed_density}
{
A \emph{tree decomposition}~$(T, \beta)$ of a graph~$G$ consists of a tree~$T$ together with a function
$\beta \colon V(T) \to 2^{V(G)}$, which assigns to each node of~$T$ a subset of vertices of~$G$, called a \emph{bag}, such that:

\begin{enumerate}
    \item For every edge $e = \{u,v\} \in E(G)$, there exists a bag $t \in V(T)$ with $u, v \in \beta(t)$.
    \item For every vertex $u \in V(G)$, the set of nodes $\{t \in V(T) \mid u \in \beta(t)\}$ induces a connected subtree of $T$.
\end{enumerate}

The \emph{width} of a tree decomposition $(T, \beta)$ is defined as
\[
\max_{t \in V(T)} |\beta(t)| - 1,
\]
and the \emph{treewidth} of $G$, denoted $\operatorname{tw}(G)$, is the minimum
width over all possible tree decompositions of $G$.
We call a tree decomposition \emph{optimal} if it attains that minimum~\cite{PARAMALG}.\\

Since dynamic programming over tree decompositions is a common technique for algorithms on graphs of bounded treewidth,
it is convenient to work with \emph{nice tree decompositions}, which have a restricted structure.
We adopt a definition similar to that of Cygan et al.~\cite{PARAMALG}:

\begin{itemize}
    \item For all leaf nodes $\ell \in V(T)$ and for the root node $r \in V(T)$,
          the corresponding bags are empty, that is, $\beta(l) = \beta(r) = \emptyset$.
    \item Every non-leaf node is of exactly one of the following three types:
    \begin{itemize}
        \item \textbf{Introduce node}: a node $t$ with exactly one child $t'$ such that
              $\beta(t) = \beta(t') \cup \{u\}$ for some vertex $u \in V(G)$.
              We say that $u$ is \emph{introduced} at $t$.

        \item \textbf{Forget node}: a node $t$ with exactly one child $t'$ such that
              $\beta(t) = \beta(t') \setminus \{u\}$ for some vertex $u \in V(G)$.
              We say that $u$ is \emph{forgotten} at $t$.

        \item \textbf{Join node}: a node $t$ with exactly two children $t_1, t_2$ such that
              $\beta(t) = \beta(t_1) = \beta(t_2)$.
    \end{itemize}
\end{itemize}

It is known that, for any integer $k$, there exists an algorithm that either computes a tree decomposition of width at most $k$ for a graph $G$, or correctly decides that $\tw(G) > k$. This algorithm runs in time $k^{\bigO(k^3)} \cdot |V(G)|$.

Furthermore, given a tree decomposition $(T, \beta)$ of $G$, one can construct a \emph{nice tree decomposition} of $G$ with at most
$\bigO(\tw \cdot |V(G)|)$ nodes in
\[
\bigO\big(\tw^2 \cdot \max(|V(T)|, |V(G)|)\big)
\]
time \cite{PARAMALG}.\\

When working over a tree decomposition, forgotten vertices still contribute
orientation weight. To keep track of these contributions, we extend the framework
to weighted fractional orientations, which explicitly store these values.

\begin{definition}
    Let $\phi$ be a fractional orientation and let $w \colon V(G) \to \QQ$ be a
	function. Then~$(\phi, w)$ is called a \emph{weighted fractional orientation}.
	Such a function $w$ is called $q$-quantized for some $q \in \NN$, if
	for every vertex $v \in V(G)$ there exists a $p \in \NN$, such that
	$w(v) = \nicefrac{p}{q}$.

    If both $w$ and $\phi$ are $q$-quantized for some $q \in \NN$,
    $(\phi, w)$ is called a \emph{$q$-quantized weighted fractional orientation}.
	Define
    	\[
        	\degexphi(v) \coloneqq \sum_{u \in N(v)} \phi(\{u, v\})^v + w(v)
        	\quad \text{and} \quad
        	\Degexphi(G) \coloneqq \max_{v \in V(G)} \degexphi(v)
    	\]
	where~$\degexphi(v)$ is the \emph{weighted indegree} of a vertex~$v \in V(G)$
	and~$\Degexphi(G)$ is the \emph{maximum weighted indegree} of~$G$ with respect to the
	weighted fractional orientation~$(\phi, w)$.
\end{definition}

\begin{definition}
	Let~$G$ be a graph and let $H \subseteq G$ be a subgraph.
	Let~$(\phi, w)$ be a weighted fractional orientation for~$H$ and let~$\phi'$
	be a fractional orientation for~$G$. Then~$(\phi, w)$ and~$\phi'$ are
	\emph{consistent} if
	\begin{align}
		\phi(e)^u = \phi'(e)^u
	\end{align}
	for all edges $e \in E(H)$ and all vertices $u \in V(H)$, and
	\begin{align}
		w(u) = \sum_{v \in N_G(v) \setminus V(H)} \phi'(\{u, v\})^u .
	\end{align}
\end{definition}

\fpttwfixeddensity*
\begin{proof}
We construct a dynamic programming algorithm.
Let $(T, \beta)$ be a nice tree decomposition of $G$.
For each bag $t \in V(T)$, let $V_t$ denote the set of vertices contained in
the bags of the subtree of $T$ rooted at $t$.

We define a dynamic programming table $c$.
Let $S \subseteq \beta(t)$ and $(\phi, w)$ be a weighted fractional orientation for $G[S]$.
Intuitively, an entry of $c$ stores the following information:
\begin{align*}
	c[t, S, (\phi, w)] =
	\begin{cases}
	    	\text{the minimum number of vertices to delete from $V_t$ in order to obtain}\\
	    	\text{a vertex set $\hat S$ with $S \subseteq \hat S \subseteq V_t$ and
		$\hat S \cap \beta(t) = S$,} \\
	    	\text{together with a fractional orientation $\hat \phi$ of $G[\hat S]$
		that is consistent} \\
	    	\text{with $(\phi, w)$ and satisfies $\Degfrac{\hat \phi} \le \taup$,}
	\end{cases}
\end{align*}
and we set $c[t, S, (\phi, w)] = \infty$ if no such solution exists.

If these conditions are satisfied, then for the root $r$ of the tree decomposition the entry
$c[r, \emptyset, (\mathbf{0}, \mathbf{0})]$ (where $\mathbf{0}$ denotes the zero function) stores the minimum number of vertices that must be deleted from~$G$
to obtain a graph that admits a fractional orientation~$\phi$ with
$\Degphi(G) \le \taup$.
By Lemma~\ref{lem:charikar}, this implies that the resulting graph has a densest subgraph with density at most~$\taup$.

We refer to $[t, S, (\phi, w)]$ as a \emph{state} of the table $c$.
If $c[t, S, (\phi, w)] = \infty$, no value is stored for this state.

\medskip

\paragraph*{Leaf node}
If $t$ is a leaf node, then the only valid state is the empty one, and we set
\[
    	c[t, \emptyset, (\mathbf{0}, \mathbf{0})] = 0.
\]

\medskip

\paragraph*{Introduce node}
Consider an introduce node $t$ that introduces a vertex $v$, and let $t'$ denote
its child.
We consider all subsets $S \subseteq \beta(t)$ and all weighted $q$-quantized
fractional orientations $(\phi, w)$ for $G[S]$.

If $v \notin S$, then $v$ is deleted in the solution. In this case, we set
\[
    c[t, S, (\phi, w)] = c[t', S, (\phi, w)] + 1.
\]

If $v \in S$, then $v$ is retained in the solution.
Let $(\phi', w')$ be the restriction of $(\phi, w)$ to $S \setminus \{v\}$, defined by
\begin{align*}
    	\phi(e)^u &= \phi'(e)^u
    	&& \text{for every edge $e \in E(G[S \setminus \{v\}])$ and
		$u \in S \setminus \{v\}$}, \\
    	w(u) &= w'(u)
    	&& \text{for every $u \in S \setminus \{v\}$}.
\end{align*}
If $\Degfrac{(\phi, w)} > \taup$, we set
\[
    	c[t, S, (\phi, w)] = \infty.
\]
Otherwise, if $\Degfrac{(\phi, w)} \le \taup$, we set
\begin{align}
    	c[t, S, (\phi, w)] = c[t', S \setminus \{v\}, (\phi', w')].
    	\label{intro:equal}
\end{align}

We now show that, assuming the values for all states at node $t'$ are minimal,
the value computed for node $t$ is also minimal.

Suppose $v \notin S$. If $c[t', S, (\phi, w)]$ is computed correctly, then so is
$c[t, S, (\phi, w)]$, since the only difference is the additional deletion of $v$.
Suppose $v \in S$. Let $\hat S$ be a set attaining the minimum in the definition of
$c[t', S \setminus \{v\}, (\phi', w')]$, and let $\hat \phi'$ be a fractional
orientation of $G[\hat S]$ consistent with $(\phi', w')$.
If $\Degfrac{(\phi, w)} > \taup$, then no fractional orientation of
$G[\hat S \cup \{v\}]$ consistent with $(\phi, w)$ can have maximum indegree at
most $\taup$. Hence, $c[t, S, (\phi, w)] = \infty$.

If $\Degfrac{(\phi, w)} \le \taup$, then a fractional orientation~$\hat \phi$ of
$G[\hat S \cup \{v\}]$ exists with
\begin{align*}
    	\hat \phi(e)^u &= \hat \phi'(e)^u
    	&& \text{for every edge $e \in E(G[\hat S])$ and
		$u \in \hat S$}
\end{align*}
and $\hat \phi$ consistent with $(\phi, w)$. Moreover,
since $N(v) \cap V_t \subseteq \beta(t)$, the indegrees of vertices in
$V_t \setminus \beta(t)$ are equal in $\hat \phi'$ and $\hat \phi$. Therefore,
\begin{align}
    	c[t, S, (\phi, w)] \le c[t', S \setminus \{v\}, (\phi', w')].
    	\label{intro:le}
\end{align}

Conversely, let $\hat S$ be a set attaining the minimum in the definition of
$c[t, S, (\phi, w)]$. Then $\hat S$ yields a feasible solution for
$c[t', S \setminus \{v\}, (\phi', w')]$, implying
\begin{align}
    	c[t, S, (\phi, w)] \ge c[t', S \setminus \{v\}, (\phi', w')].
    	\label{intro:ge}
\end{align}

Together, inequalities \eqref{intro:le} and \eqref{intro:ge} establish
\eqref{intro:equal} for the case $v \in S$.

\medskip

\paragraph*{Forget node}
Consider a forget node $t$ that forgets a vertex $v$, and let $t'$ denote its
child.
We consider all subsets $S \subseteq \beta(t)$ and all weighted fractional
orientations $(\phi, w)$ for~$G[S]$.

Let $\Phi$ be the set of all weighted fractional orientations for
$G[S \cup \{v\}]$ such that for each $(\phi', w') \in \Phi$ the following holds:
\begin{align}
    	\phi(e)^u &= \phi'(e)^u
    	&& \text{for all edges $e \in E(G[S])$ and all $u \in S$,}
    	\label{forget:edge_equivalents} \\
    	w(u) &= w'(u) + \phi'(\{u, v\})^u
    	&& \text{for all $u \in S$.}
    	\label{forget:weight_equivalents}
\end{align}
We define
\begin{align}
    	c[t, S, (\phi, w)] =
    	\min \left\{
        	\min_{(\phi', w') \in \Phi}
        	c[t', S \cup \{v\}, (\phi', w')],
        	\;
        	c[t', S, (\phi, w)]
    	\right\}.
    	\label{forget:eq}
\end{align}
We now show that, assuming the values for all states at node $t'$ are minimal,
the value computed for node $t$ is also minimal.

Let $\hat S$ be a set attaining the minimum in the definition of
$c[t, S, (\phi, w)]$, and let $\hat \phi$ be a fractional orientation of
$G[\hat S]$ consistent with $(\phi, w)$ and satisfying
$\Degfrac{\hat \phi} \le \taup$.

If $v \notin \hat S$, then $\hat S$ and $\hat \phi$ are considered in the
definition of $c[t', S, (\phi, w)]$, and hence
\[
    	c[t, S, (\phi, w)] \ge c[t', S, (\phi, w)].
\]

If $v \in \hat S$, then $\hat S$ appears in the definition of
$c[t', S \cup \{v\}, (\phi', w')]$ for some $(\phi', w') \in \Phi$ consistent
with $\hat \phi$, implying
\[
    	c[t, S, (\phi, w)] \ge c[t', S \cup \{v\}, (\phi', w')].
\]

Combining both cases yields
\begin{align}
    	c[t, S, (\phi, w)] \ge
    	\min \left\{
        	\min_{(\phi', w') \in \Phi}
        	c[t', S \cup \{v\}, (\phi', w')],
        	\;
        	c[t', S, (\phi, w)]
    	\right\}.
    	\label{forget:ge}
\end{align}

Conversely, any set $\hat S$ and fractional orientation $\hat \phi$ considered in
the definition of $c[t', S, (\phi, w)]$ are also considered in the definition of
$c[t, S, (\phi, w)]$. Thus,
\begin{align}
    	c[t, S, (\phi, w)] \le c[t', S, (\phi, w)].
    	\label{forget:le_v_not_in_S}
\end{align}

Now let $\hat S \subseteq V_{t'}$ with $v \in \hat S$, and let $\hat \phi$ be a
fractional orientation for $G[\hat S]$ considered in the definition of
$c[t', S \cup \{v\}, (\phi', w')]$.
Here, $S \cup \{v\} = \hat S \cap \beta(t')$, and $(\phi', w')$ is consistent with
$\hat \phi$.
If $(\phi, w)$ is any weighted fractional orientation extending $(\phi', w')$
via \eqref{forget:edge_equivalents} and \eqref{forget:weight_equivalents}, then
$\hat S$ and $\hat \phi$ are also considered in the definition of
$c[t, S, (\phi, w)]$. Hence,
\[
    	c[t, S, (\phi, w)] \le c[t', S \cup \{v\}, (\phi', w')].
\]
Since this holds for all $(\phi', w') \in \Phi$, we obtain
\begin{align}
    	c[t, S, (\phi, w)] \le
    	\min_{(\phi', w') \in \Phi}
    	c[t', S \cup \{v\}, (\phi', w')].
    	\label{forget:le_v_in_S}
\end{align}

Combining \eqref{forget:le_v_not_in_S} and \eqref{forget:le_v_in_S}, we arrive at
\begin{align}
    	c[t, S, (\phi, w)] \le
    	\min \left\{
        	\min_{(\phi', w') \in \Phi}
        	c[t', S \cup \{v\}, (\phi', w')],
        	\;
        	c[t', S, (\phi, w)]
    	\right\}.
    	\label{forget:le}
\end{align}

Together, inequalities \eqref{forget:ge} and \eqref{forget:le} establish the
desired equality \eqref{forget:eq}.

\paragraph*{Join node}
Consider a join node $t$ and let $t_1$ and $t_2$ denote its children.
Let $S \subseteq \beta(t) = \beta(t_1) = \beta(t_2)$ be arbitrary, and let
$(\phi, w)$ be a weighted fractional orientation for~$G[S]$.

Let $\Phi$ be the set of all pairs of weighted fractional orientations
$((\phi_1, w_1), (\phi_2, w_2))$ for~$G[S]$ such that the following holds:
\begin{align}
    	\phi(e)^u &= \phi_1(e)^u = \phi_2(e)^u
    	&& \text{for all edges $e \in E(G[S])$ and all $u \in S$,}
    	\label{join:edge_equality} \\
    	w(u) &= w_1(u) + w_2(u)
    	&& \text{for all $u \in S$.}
    	\label{join:weight_equality}
\end{align}
We define
\begin{align}
    	c[t, S, (\phi, w)] =
    		\min_{((\phi_1, w_1), (\phi_2, w_2)) \in \Phi}
    		\bigl(
        		c[t_1, S, (\phi_1, w_1)]
        		+
        		c[t_2, S, (\phi_2, w_2)]
        		-
        		(|\beta(t)| - |S|)
    		\bigr).
    	\label{join:eq}
\end{align}
We now show that, assuming the values for all states of $t_1$ and $t_2$ are
minimal, the value computed for $t$ is also minimal.

Let $\hat S$ be a set attaining the minimum in the definition of
$c[t, S, (\phi, w)]$, and define
$\hat S_1 = \hat S \cap V_{t_1}$ and
$\hat S_2 = \hat S \cap V_{t_2}$.
Then
\[
    c[t_1, S, (\phi_1, w_1)] \le |V_{t_1}| - |\hat S_1|
    \quad \text{and} \quad
    c[t_2, S, (\phi_2, w_2)] \le |V_{t_2}| - |\hat S_2|.
\]

Since $\hat S_1 \cap \hat S_2 = S$, we obtain
\begin{align}
\begin{split}
    	c[t, S, (\phi, w)]
    	&=
    	|V_{t_1}| - |\hat S_1|
    	+
    	|V_{t_2}| - |\hat S_2|
    	-
	(|\beta(t)| - |S|) \\
	&\ge
	c[t_1, S, (\phi_1, w_1)]
	+
	c[t_2, S, (\phi_2, w_2)]
	-
	(|\beta(t)| - |S|).
\end{split}
\label{join:ge}
\end{align}

Conversely, let $((\phi_1, w_1), (\phi_2, w_2)) \in \Phi$
be a pair attaining the minimum in \eqref{join:eq}.
Let~$\hat S_1 \subseteq V_{t_1}$ and~$\hat S_2 \subseteq V_{t_2}$ be sets attaining
the minima in the definitions of
$c[t_1, S, (\phi_1, w_1)]$ and $c[t_2, S, (\phi_2, w_2)]$, respectively.
Let $\hat \phi_1$ and $\hat \phi_2$ be fractional orientations of~$G[\hat S_1]$
and~$G[\hat S_2]$ consistent with $(\phi_1, w_1)$ and
$(\phi_2, w_2)$, respectively.

Define $\hat S = \hat S_1 \cup \hat S_2$.
Since $\hat S_1 \cap \hat S_2 = S$, it follows that
$\hat S \cap \beta(t) = S$.
Moreover,~$\hat \phi_1$ and~$\hat \phi_2$ agree on all edges of~$G[S]$ by~\eqref{join:edge_equality},
and thus together define a fractional orientation~$\hat \phi$ of~$G[\hat S]$
that is consistent with~$(\phi, w)$.

The indegree of each vertex $u \in \hat S$ under $\hat \phi$ is the sum of its
indegrees under $\hat \phi_1$ and $\hat \phi_2$ minus the contribution from edges
entirely contained in $G[S]$.
By \eqref{join:weight_equality}, this is exactly captured by $w(u)$, and hence
$\Degfrac{\hat \phi} \le \taup$.

Therefore,
\begin{align}
\begin{split}
    c[t, S, (\phi, w)]
    &\le
    c[t_1, S, (\phi_1, w_1)]
    +
    c[t_2, S, (\phi_2, w_2)]
    -
    (|\beta(t)| - |S|).
\end{split}
\label{join:le}
\end{align}

Combining \eqref{join:ge} and \eqref{join:le} establishes \eqref{join:eq}.

\bigskip

\paragraph*{Running time}
We first examine the size of the dynamic programming table $c$.
For each bag $t$, we consider all subsets $S \subseteq \beta(t)$ and,
for each such subset, all possible weighted fractional orientations.
The number of edges in $G[\beta(t)]$ is at most $\tw^2$, and each edge
has~$q+1$ possible orientations due to $q$-quantization.
Considering the weights, each vertex in $S$ can be assigned a value between~$0$
and~$\taup$, giving at most~$p+1$ possible weights per vertex.
Hence, the number of possible weighted fractional orientations per subset is bounded by
\[
    	(q+1)^{\tw^2} \cdot (p+1)^{\tw+1} \le ((q+1) \cdot (p+1))^{\tw^2+1}.
\]
Since there are at most $2^{\tw+1}$ subsets, the table size for each bag is bounded by
\[
    	2^{\tw+1} \cdot ((q+1) \cdot (p+1))^{\tw^2+1} \le (2 \cdot (q+1) \cdot (p+1))^{\tw^2+1}.
\]

The most time-consuming operation is the minimization over all pairs of
weighted fractional orientations within a bag, which is quadratic in the table size.
Therefore, the overall running time of the algorithm is
\[
    	(2 \cdot (q+1) \cdot (p+1))^{\bigO(\tw^2)} \cdot n^{\bigO(1)}.
\]

This concludes the proof of the first stated running time.
The second stated running time can be derived as follows.

\begin{lemma}[Bazgan et al.~\cite{BNV25}]
	Let $G^*$ be a densest subgraph of $G$ with $\rho(G^*) = \rho^*(G)$.
	Then the minimum degree in $G^*$ is at least $\lceil \rho(G^*) \rceil$.
	\label{lem:min_degree_densest_subgraph}
\end{lemma}

For any graph~$G$ it holds that $\rho^*(G) \le \tw(G)$.
To show this let $G^*$ be a densest subgraph of $G$.
Clearly, $\tw(G^*) \le \tw(G)$ holds and hence there exists a vertex~$v \in V(G^*)$ with $\deg_{G^*}(v) \le \tw(G^*) \le \tw(G)$.
With Lemma~\ref{lem:min_degree_densest_subgraph},
this implies that~$\rho(G^*) \leq \tw(G)$.

Hence, if $\taup \ge \tw$, the problem is trivial.
We may therefore assume $\taup < \tw$, and consequently $p < \tw \cdot q$.
Applying the algorithm above, the runtime is then bounded by~$(2 (q+1) \cdot (\tw \cdot q + 1))^{\bigO(\tw^2)} \cdot n^{\bigO(1)}$.

\end{proof}
}

\subsection{Vertex Integrity}
In the previous subsections we established fixed-parameter tractability of \bdvd for restricted target density.
As we will see in \cref{sec:hardness} the problem becomes much harder in case we allow arbitrary target densities.
For this setting we thus turn to the larger parameter vertex integrity~$\vi$ of the input graph~$G$ and show that \bdvd is in~$\FPT$ parameterized by~$\vi$.
The \emph{vertex integrity} of a graph~$G$ is the smallest number~$\vi$ of vertices such that we can remove~$\vi$ vertices and every component in the obtained graph has size at most~$\vi$.
Note that such a witness set of at most~$\vi$ vertices can be computed in FPT time \cite{DrangeDH16}.
The parameter~$\vi$ lies between the vertex cover number and the treedepth \cite{GimaHKKO22}.
In this sense our results improves a result from Bazgan et al.~\cite{BNV25}, who showed that \bdvd is in FPT for the vertex cover number, and provides a sharp bound for the tractability in the current landscape of structural parameters.

\begin{theorem}
\label{thm:vertex_integrity}
	Given a witness for the vertex integrity, \BDVD{} can be solved in~$2^{\vi \cdot 2^{O(\vi^2)}}n^{O(1)}$.
\end{theorem}

\begin{proof}
 Let~$(G, k, \taup = \nicefrac{p}{q})$ be a \bdvd-instance and~$|X| \leq \vi$ a set of vertices such that every component of~$G-X$ has size at most~$\vi$.
 By branching over all possible subsets of~$X$, we can
 guess which vertices of~$X$ we do and do not choose into the solution.
 For the rest of the proof we therefore assume no vertex from the set~$X$ is deleted, by considering only vertices not taken into the solution by the branching.
 The idea of the algorithm is to formulate an ILP which decides how many components with the same 'type' are transformed in the same way by deleting the vertices of a solution.

 We say two components~$C_1$ and~$C_2$ of~$G-X$ have the same \emph{type}~$t$ if there exists a bijection~$\phi : V(C_1) \rightarrow V(C_2)$ such that for all~$u \in X, v,v' \in V(C_1)$
 \begin{itemize}
	\item $uv \in E(G)$ if and only if~$u\phi(v) \in E(G)$ and
	\item $vv' \in E(G)$ if and only if~$\phi(v)\phi(v') \in E(G)$.
 \end{itemize}
 The number of different types is bounded by~$\vi \cdot 2^{2\vi^2}$.

 For every component type~$t$ of~$G-X$ and for every component type~$s$ which can be obtained by deleting vertices from a component with type~$t$, we introduce one variable~$x_{t}^s$.
 Intuitively, the value of the variable corresponds to the number of components with type~$t$ which are turned into components with type~$s$ when deleting the vertices of a solution.
 If type~$s$ can be obtained from~$t$, then we write~$s \subseteq t$.
 Further, we denote by~$k_t^s$ the number of vertices which need to be deleted in this process and by~$n_t$ the number of components with type~$t$ in~$G-X$.
 This allows us to formulate the following objective function and the first constraints:
 \begin{align}
	\min \quad & \sum_{t} \sum_{s, s \subseteq t} k_t^sx_t^s \notag \\
	\text{s.t.} \quad
	& \sum_{s, s\subseteq t} x_{t}^s = n_t \quad && \forall \text{ types } t \text{ of } G-X.\label{cons:component_numbers} \\
	& 0 \leq x_t^s \quad && \forall \text{ types } t,s \text{ of } G-X.\label{cons:component_numbers}
 \end{align}

 Let~$G'$ be the graph obtained from a solution satisfying the above constraints, meaning~$G'$ is the subgraph of~$G$ which contains~$X$ and for any type~$s$, the graph~$G'-X$ has~$\sum_{t, s\subseteq t} x_t^s$ components of type~$s$.
 It remains to add constraints which ensure $G'$ contains no subgraph with density larger than~$\taup = \nicefrac{p}{q}$.
 \begin{claim}
  Suppose~$H$ is a densest subgraph in~$G'$ and let~$C$ and~$C'$ be two components of the same type in~$G'-X$.
  We can then assume that~$H$ restricted to~$C$ and~$C'$ yields again two components with the same type.
 \end{claim}
 \begin{claimproof}
  Since~$C$ and~$C'$ are of the same type in~$G'-X$, there exists a bijection~$\phi:V(C)\to V(C')$.
  Without loss of generality assume~$\rho(H[V(C) \cup X])\geq\rho(H[V(C') \cup X])$.
  Then updating~$H$ such that for any~$v \in V(C)$ it holds that~$\phi(v) \in V(H) \iff v \in V(H)$¸ cannot decrease~$\rho(H)$.
 \end{claimproof}

 As a consequence we only need to care about the density of subgraphs which behave the same on all components with the same type.
 This allows us to formulate density constraints that ensure that in each remaining subgraph~$H$ we will have~$|E(H)| / |V(H)| \le \taup$, or equivalently~$q\cdot|E(H)| \le p \cdot |V(H)|$.
 A subtype assignment~$w$ assigns each type~$s$ one type~$w(s)$ such that~$w(s) \subseteq s$; the idea is that~$w(s)$ indicates the part of~$s$ that are in the subgraph~$H$.
 Let~$W \subseteq G[X]$.
 We denote the number of vertices in a component of type~$w(s)$ by~$n_{w(s)}$ and the number of edges in the component and between the component and~$W$ by~$m_{w(s)}^W$.
 Finally, for each subgraph~$W \subseteq G[X]$ and subtype assignment~$w$ we introduce the following constraint that ensures~$q\cdot|E(H)| \le p \cdot |V(H)|$ for the corresponding subgraph:

 \begin{align}
	q\cdot\big( |E(W)| + \sum_s \sum_{s\subseteq t} m_{w(s)}^W x_t^s  \big) \leq  p\cdot\big( |V(W)| + \sum_s \sum_{s\subseteq t} n_{w(s)} x_t^s  \big). \label{cons:component_numbers}
 \end{align}

 The ILP has~$2^{O(\vi^2)}$ many variables and~$2^{O(\vi^2)}$ many constraints.
 Thus, the ILP can be solved in~$2^{\vi \cdot 2^{O(\vi^2)}} n^{O(1)}$ time \cite{FrankT87, Kannan87, Lenstra83}.
\end{proof}

\subsection{Max Leaf Number}
The \emph{max leaf number}~$\ml$ of a graph~$G$ is the maximum number of leaves in any spanning tree of~$G$.
It is a rather large parameter since on graphs with bounded max leaf number the number of high-degree vertices is also bounded.
To be more precise, let~$V_{\geq 3}$ be the set of vertices with degree at least three in~$G$, then the size of~$V_{\geq 3}$ is bounded by~$O(\ml^2)$ \cite{Eppstein15, GarvardtK24}.
The next theorem uses this fact to show that \BDVD is in FPT with respect to the max leaf number.
In case the target density is greater than one, it is sufficient to branch on all possible sets of vertices in~$V_{\geq 3}$ to take into a solution.
In case the target density is below one, after guessing which vertices form~$V_{\geq 3}$ to take into a solution and how to cut the paths between them we require an ILP to calculate whether all resulting connected components can be made small enough by well-placed cuts.

\begin{restatable}[\appref{thm:max_leaf_number}]{theorem}{maxleafnumber}
\label{thm:max_leaf_number}
	\BDVD can be solved in~$\ml^{O(\ml^2)} \cdot n^{O(1)}$ running time.
\end{restatable}

\appendixproof{thm:max_leaf_number}{
\maxleafnumber*
\begin{proof}
	Let~$(G,k, \taup)$ be a \BDVD-instance.
	We distinguish between the cases~$\taup > 1$ and~$\taup < 1$ (and note that~$\taup=1$ is covered by \cref{thm:fpt_cw_fixed_density}).

	$\taup > 1$: Since any densest subgraph~$H$ of~$G$ has minimum degree at least~$ \lceil \rho(H) \rceil$ \cite{BNV25}, we can iteratively remove all degree one vertices and assume that the minimum degree in~$G$ is two.
	We show that we can restrict ourselves to delete vertices from~$V_{\geq 3}$, which directly implies a brute force algorithm running in~$2^{O(\ml^2)} \cdot n^{O(1)}$ time.
	Suppose there is a set of vertices~$X$ with~$\rho^*(G-X) \leq \taup$ such that~$X$ contains a vertex~$v \notin V_{\geq 3}$.
	The vertex~$v$ lies on a path~$P$ between two vertices~$u$ and~$w$ from~$V_{\geq 3}$.
	Consider~$X' := (X \setminus \{v\}) \cup \{w\}$ and assume towards a contradiction that~$G-X'$ contains a densest subgraph~$H$ with density greater than~$\taup$.
	Since the minimum degree in~$H$ is at least two, we conclude that~$H$ contains no vertex from~$P\setminus \{u,w\}$.
	Thus~$H$ would also be too dense in~$G-X$ which is a contradiction.

	$\taup < 1$: W.l.o.g. the target density is~$\taup = 1-\nicefrac{1}{q}$ \cite{NICH26}.
	This means we are looking for a set~$X$ of vertices such that each component of~$G-X$ is a tree on at most~$q$ vertices. 
	In the first step we get rid of all degree one vertices.
	For this, we initialize a counter~$c_v$ for every vertex~$v$ with an initial value of one and from now on, think about the following problem.
	\problemDef{Annotated Bounded Tree Vertex Deletion}
{
	A graph $G$, a number~$c_v$ for every vertex~$v$ and two integers $k,q \in \NN$.
}
{
	Is there a subset $S \subseteq V(G)$ with $|S| \le k$ such that each component~$C$ of~$G-S$ is a tree with~$\sum_{v \in V(C)} c_v \leq q$?
}
	Intuitively, the counter~$c_v$ keeps track of how many vertices 'hide' behind the vertex~$v$ including itself.
	By exhaustively applying the following reduction rule, we are able to remove all degree one vertices from the graph.
	\begin{reduction}
	\label{red_rule_1}
		Le~$u$ be a degree one vertex with neighbor~$v$.
		Remove~$u$ from the graph and increase the counter~$c_v$ by the value~$c_u$.
		If~$c_v$ is greater than~$q$ afterwards, then add~$v$ to~$X$ and remove it from~$G$.
	\end{reduction}
	In the next step, we try every possible subset~$V'$ of the remaining vertices from~$V_{\geq 3}$\footnote{The set~$V_{\geq 3}$ will always denote vertices of the original graph (before any data reduction) of degree at least three.} as a part of the solution~$X$.
	For every such subset~$V'$ we proceed as follows.
	We add~$V'$ to~$X$, remove~$V'$ from the graph and again exhaustively apply \cref{red_rule_1} to get rid of all degree one vertices.
	If~$|X|$ becomes greater than~$k$, then the solution can be discarded and we can continue with the next set~$V'$.
	Otherwise we enter the final phase in which we check via an ILP if we can delete internal vertices from the remaining paths such that each component is a tree of size at most~$q$.

	Note that all vertices~$v$ in the remaining graph with~$c_v > 1$ are vertices from the set of vertices with degree at least three in~$G$, that is, from the set~$V_{\geq 3}$.
	We try every possible subset~$\mathcal{P}$ of paths between these vertices as paths from which the solution contains at least one vertex.
	There are at most~$2^{O(\ml^2)}$ such sets~$\mathcal{P}$ as~$\sum_{v \in V_{\geq 3}} \deg(v) \leq 42 \ml^2$ \cite{GarvardtK24}.
	Let~$H_{\mathcal{P}}$ be the graph which is obtained by deleting every internal vertex from all paths of~$\mathcal{P}$ and let~$\mathcal{C_\mathcal{P}}$ be the set of components of~$H_{\mathcal{P}}$.
	If any component from~$\mathcal{C_\mathcal{P}}$ contains a cycle or has size greater than~$q$, then the current solution can be discarded and we continue with the next set of paths.
	Otherwise, it remains to check if vertices from the paths can be added to~$X$ in such a way that every component of~$G-X$ has size at most~$q$.
	We formulate this problem as an ILP feasibility problem.
	For every path~$P \in \mathcal{P}$ we fix one endpoint as the left endpoint~$l_P$ and one endpoint as the right endpoint~$r_P$.
	The idea is to introduce one variable~$x_P$ for every path~$P \in \mathcal{P}$ where the~$x_P$-th vertex of~$P$ starting from the neighbor of~$l_P$ should be contained in the solution~$X$.
	Since every component has size at most~$q$, we then assume that also the~$x_P + q + 1$-th vertex, the~$x_P + 2(q + 1)$-th vertex and so on are contained in the solution.
	In particular, given the value of~$x_P$, we can derive how many vertices are added to the components of~$l_P$ and~$r_P$ from~$\mathcal{C_\mathcal{P}}$ respectively.
	To be able to derive these numbers in the final ILP, we need to distinguish between different kinds of paths in~$\mathcal{P}$ depending on their length.

	We denote the number of vertices in the path~$P \in \mathcal{P}$ by~$n_P$ and distinguish between two cases in the first step.
	\begin{itemize}
	 \item $\mathcal{P}_{\ell\cdot(q+1)}$ contains the paths~$P$ for which~$n_P$ is a multiple of~$q+1$.
	 \item $\mathcal{P}_{R}$ contains the rest of the paths.
	\end{itemize}
	If~$P \in \mathcal{P}_{\ell\cdot(q+1)}$ with~$n_P = \ell_P (q+1)$ for some number~$\ell_P$, then the solution contains exactly~$\ell_P$ vertices from~$P$ as we only add internal vertices of the paths to the solution in this phase.
	However for most of the paths in~$\mathcal{P}_{R}$ there are two possibilities which need to be considered.
	\begin{figure}[t]
    \centering
    \begin{tikzpicture}[xscale=0.66, yscale=0.44]

        \node[] () at (-5,7) {\textcolor{blue}{$l_P$}};
        \node[circle, fill=blue, inner sep=3.25pt] (l2) at (-5,6) {};
        \node[circle, fill=black, inner sep=2pt] (v21) at (-4,6) {};
        \node[circle, fill=black, inner sep=2pt] (v22) at (-3,6) {};
        \node[] () at (-2,5) {\textcolor{gray}{$x_P = 3$}};
        \node[circle, fill=gray, inner sep=2pt] (v23) at (-2,6) {};
        \node[circle, fill=black, inner sep=2pt] (v24) at (-1,6) {};
        \node[circle, fill=black, inner sep=2pt] (v25) at (0,6) {};
        \node[circle, fill=black, inner sep=2pt] (v26) at (1,6) {};
        \node[circle, fill=black, inner sep=2pt] (v27) at (2,6) {};
        \node[] () at (3,5) {\textcolor{gray}{$x_P + q + 1$}};
        \node[circle, fill=gray, inner sep=2pt] (v28) at (3,6) {};
        \node[] () at (4,7) {\textcolor{red}{$r_P$}};
        \node[circle, fill=red, inner sep=3.25pt] (r2) at (4,6) {};
        \node[] () at (-7,6) {$\mathcal{P}_{\ell \cdot (q+1)}$};
        \draw[thick, gray, dashed] (v22)--(v23)--(v24) (v27)--(v28)--(r2);
        \draw[very thick] (l2)--(v21)--(v22) (v24)--(v25)--(v26)--(v27);
        \node[] () at (7,6) {$n_P = 10$};

        \node[] () at (-4,3) {\textcolor{blue}{$l_P$}};
        \node[circle, fill=blue, inner sep=3.25pt] (l3) at (-4,2) {};
        \node[] () at (-3,1) {\textcolor{gray}{$x_P = 1$}};
        \node[circle, fill=gray, inner sep=2pt] (v31) at (-3,2) {};
        \node[circle, fill=black, inner sep=2pt] (v32) at (-2,2) {};
        \node[circle, fill=black, inner sep=2pt] (v33) at (-1,2) {};
        \node[circle, fill=black, inner sep=2pt] (v34) at (0,2) {};
        \node[circle, fill=black, inner sep=2pt] (v35) at (1,2) {};
        \node[] () at (2,1) {\textcolor{gray}{$x_P + q + 1$}};
        \node[circle, fill=gray, inner sep=2pt] (v36) at (2,2) {};
        \node[] () at (3,3) {\textcolor{red}{$r_P$}};
        \node[circle, fill=red, inner sep=3.25pt] (r3) at (3,2) {};
        \node[] () at (-7,2) {$\mathcal{P}_{\lceil R \rceil}$};
        \draw[thick, gray, dashed] (l3)--(v31)--(v32) (v35)--(v36)--(r3);
        \draw[very thick] (v32)--(v33)--(v34)--(v35);
        \node[] () at (7,2) {$n_P = 8$};

        \node[] () at (-4,-1) {\textcolor{blue}{$l_P$}};
        \node[circle, fill=blue, inner sep=3.25pt] (l4) at (-4,-2) {};
        \node[circle, fill=black, inner sep=2pt] (v41) at (-3,-2) {};
        \node[circle, fill=black, inner sep=2pt] (v42) at (-2,-2) {};
        \node[circle, fill=black, inner sep=2pt] (v43) at (-1,-2) {};
        \node[] () at (0,-3) {\textcolor{gray}{$x_P = 4$}};
        \node[circle, fill=gray, inner sep=2pt] (v44) at (0,-2) {};
        \node[circle, fill=black, inner sep=2pt] (v45) at (1,-2) {};
        \node[circle, fill=black, inner sep=2pt] (v46) at (2,-2) {};
        \node[] () at (3,-1) {\textcolor{red}{$r_P$}};
        \node[circle, fill=red, inner sep=3.25pt] (r4) at (3,-2) {};
        \node[] () at (-7,-2) {$\mathcal{P}_{\lfloor R \rfloor}$};
        \draw[thick, gray, dashed] (v43)--(v44)--(v45);
        \draw[very thick] (l4)--(v41)--(v42)--(v43) (v45)--(v46)--(r4);
        \node[] () at (7,-2) {$n_P = 8$};

    \end{tikzpicture}
    \caption{A sketch of the different paths that need to be considered for the case~$\taup = \nicefrac{q-1}{q} < 1$ in the proof of \Cref{thm:max_leaf_number}
    Here~$q = 4$.
    Gray vertices are part of a potential solution and the value of the variable~$x_P$ determines the first vertex which is added to the solution starting from the neighbor of~$l_P$.
    There are three cases of how many vertices are added to the component of~$r_P$; (i) $\mathcal{P}_{\ell \cdot (q+1)}$, if the number of vertices~$n_P$ is a multiple of~$q+1$, then~$q - 1 - x_P$ vertices are added; (ii) $\mathcal{P}_{\lceil R \rceil}$, we are in none of the previous cases and the solution contains~$\lceil \nicefrac{n_P}{q+1} \rceil$ vertices from the path and then~$(n_P \mod (q+1)) - x_P - 2$ vertices are added; (iii) $\mathcal{P}_{\lfloor R \rfloor}$, we are in none of the previous cases and the solution contains~$\lfloor \nicefrac{n_P}{q+1} \rfloor$ vertices from the path and then~$(n_P \mod (q+1)) - x_P - 1 + q$ vertices are added.}
    \label{fig:ml_paths}
\end{figure}
	This is sketched in parts of~\cref{fig:ml_paths}.
	For~$P \in \mathcal{P}_{R}$, the solution either contains~$\lceil \nicefrac{n_P}{q+1} \rceil$ or~$\lfloor \nicefrac{n_P}{q+1} \rfloor$ vertices.
	The problem of splitting~$\mathcal{P}_{R}$ into these two cases can be resolved by branching on a partition of~$\mathcal{P}_{R}$ into two sets~$\mathcal{P}_{\lceil R \rceil} $ and~$ \mathcal{P}_{\lfloor R \rfloor}$ where the solution contains~$\lceil \nicefrac{n_P}{q+1} \rceil$ vertices from each path in~$\mathcal{P}_{\lceil R \rceil} $ and~$\lfloor \nicefrac{n_P}{q+1} \rfloor$ vertices from each path in~$ \mathcal{P}_{\lfloor R \rfloor}$.
	One thing we need to be careful about is the fact that not all paths from~$\mathcal{P}_{R}$ can be part of both~$\mathcal{P}_{\lceil R \rceil} $ and~$ \mathcal{P}_{\lfloor R \rfloor}$.
	There are three exceptions.
	\begin{itemize}
	 \item If~$n_P < q+1$, then~$P$ needs to be contained in~$\mathcal{P}_{\lceil R \rceil} $ as in the current branch of selected paths~$\mathcal{P}$ the solution needs to delete at least one vertex from every path in~$\mathcal{P}$.
	 \item If~$n_P \mod (q+1) = 1$, then the case of deleting~$\lceil \nicefrac{n_P}{q+1} \rceil$ vertices from~$P$ is equivalent to the case where also~$l_P$ and~$r_P$ are deleted since we can add~$l_P$ and every~$q+1$-th vertex to the solution if the solution contains~$\lceil \nicefrac{n_P}{q+1} \rceil$ vertices from~$P$.
	 This is considered in a different branch where another set~$V'$ of~$V_{ \geq 3}$ is regarded as part of the solution.
	 \item If~$n_P \mod (q+1) = 2$, then the case of deleting~$\lceil \nicefrac{n_P}{q+1} \rceil$ vertices from~$P$ is equivalent to the case where also~$l_P$ or~$r_P$ is deleted by a similar reason as above.
	 This is again considered in a different branch for another subset~$V'$ of~$V_{ \geq 3}$.
	\end{itemize}

	After splitting the set~$\mathcal{P}_{R}$ into~$\mathcal{P}_{\lceil R \rceil} $ and~$ \mathcal{P}_{\lfloor R \rfloor}$ while taking the three mentioned exceptions into account, we know exactly how many vertices a potential solution contains, namely the size of~$X$ plus the sum of all vertices which need to be added for the different paths in~$\mathcal{P}$.
	If this number is greater than~$k$, then the current solution can be discarded.
	Otherwise it remains to check if the value of the variables~$x_P$ can be set in such a way that all resulting components have size at most~$q$.

	To formulate this as an ILP feasibility problem we need to distinguish between the different kinds of paths as the exact constraints depend on this.
	The component of~$l_P$ always receives~$x_P -  1$ vertices but the value range of~$x_P$ as well as the number of vertices which are added to the component of~$r_P$ depend on what kind of path we are considering.
	For an overview see~\cref{fig:ml_paths}.
	\begin{itemize}
	 \item For~$P \in \mathcal{P}_{\ell\cdot(q+1)}$ recall that in this step of the algorithm we only delete internal vertices from the paths.
	 Thus we have the constraint
	 \setcounter{equation}{0}
	 \begin{align}
		1 \leq x_P < q
	 \end{align}
	 as the case~$x_P = q$ is equivalent of adding~$r_P$ to the solution but this is considered in a different branch from earlier.
	 Here, the component of~$r_P$ receives~$q - 1 - x_P$ vertices.
	 This is important for the last constraint which we will state later.
	 \item For~$P \in \mathcal{P}_{\lceil R \rceil}$ recall that~$n_P \mod (q+1) \geq 2$ due to the above case and the exceptions mentioned before. We add the constraint
	 \begin{align}
		1 \leq x_P \leq (n_P \mod (q+1)) - 2
	 \end{align}
	 as for greater values of~$x_P$ a solution only needs to contain~$\lfloor \nicefrac{n_P}{q+1} \rfloor$ vertices from~$P$ contradicting that~$P \in \mathcal{P}_{\lceil R \rceil}$.
	 The component of~$r_P$ receives~$(n_P \mod (q+1)) - 2 - x_P$ vertices in this case.
	 \item For~$P \in \mathcal{P}_{\lfloor R \rfloor}$ we have the constraint
	 \begin{align}
		n_P \mod (q+1) \leq x_P \leq q
	 \end{align}
	 as for smaller values of~$x_P$ a solution either is equivalent to adding~$r_P$ to a solution which is considered in a different branch from earlier or needs to contain at least~$\lceil \nicefrac{n_P}{q+1} \rceil$ vertices from~$P$ which contradicts~$P \in \mathcal{P}_{\lfloor R \rfloor}$.
	 In this case~$(n_P \mod (q+1)) - x_P - 1 + q$ vertices are added to the component of~$r_P$.
	\end{itemize}

	It remains to formulate a constraint for each component~$C \in \mathcal{C_\mathcal{P}}$ which ensures that the component is not becoming too large when adding the respective vertices of the paths from~$\mathcal{P}$ to the solution.
	For every component~$C \in \mathcal{C_\mathcal{P}}$ let~$c_C$ be the sum of all counters from the component, that is, $c_C = \sum_{v \in V(C)} c_v$.
	Intuitively, $c_C$ is the actual size of the component~$C$ if already removed vertices are taken into account.
	Thus, no more than~$q - c_C$ vertices from the paths in~$\mathcal{P}$ should be added to the component~$C$.
	This is captured by the following constraint.
	\begin{align}
		& \sum_{\substack{P \in\mathcal{P} \\ l_P \in V(C)}} (x_P - 1)
		+ \sum_{\substack{P \in \mathcal{P}_{\ell\cdot(q+1)} \\ r_P \in V(C)}} (q - 1 - x_P) \notag \\
		& + \sum_{\substack{P \in \mathcal{P}_{\lceil R \rceil} \\ r_P \in V(C)}} ((n_P \mod (q+1)) - x_P - 2) + \sum_{\substack{P \in \mathcal{P}_{\lfloor R \rfloor} \\ r_P \in V(C)}} ((n_P \mod (q+1)) - x_P - 1 + q) \notag \\ & \leq q - c_C
	\end{align}
	If the ILP containing all the constraints from above is feasible, then we know that we can add vertices from the paths in~$\mathcal{P}$ to~$X$ in such a way that no component of~$G-X$ has size more than~$q$.
	If we reach this point, then we can return yes.

	The ILP has~$O(\ml^2)$ many variables.
	Thus, the it can be solved in~$\ml^{O(\ml^2)} \cdot n^{O(1)}$ time \cite{FrankT87, Kannan87, Lenstra83} and therefore the overall running time of the algorithm is bounded by~$\ml^{O(\ml^2)} \cdot n^{O(1)}$.
\end{proof}
}

We remark that the case of $\taup > 1$ from the algorithm above actually shows fixed-parameter tractability for the smaller parameter feedback edge number.
After removing all degree-one vertices, the number of vertices with degree at least three is bounded by a function of the feedback edge number.
In contrast to this, the case of~$\taup < 1$ is unclear regarding the feedback edge number.
It seems plausible that \BDVD is~$\W[1]$-hard for this parameter since the same is true for the edge deletion variant~\cite{BNV25}.

\section{Hardness Results}
\label{sec:hardness}
Here, we contrast our algorithmic findings with hardness results.
In particular, we provide two reductions that show W[1]-hardness of \bdvd{} with respect to treedepth combined with further parameters.
This implies hardness for smaller parameters such as treewidth and cliquewidth for which we previously established fixed-parameter tractability when combined with~$p$ and~$q$, the ``encoding size'' of~$\taup = \nicefrac{p}{q}$.

\subsection
{Treedepth, Distance to Cograph, and Distance to Interval Graph}
In the previous section we established that \bdvd{} is fixed-parameter tractable with respect to vertex integrity.
We contrast this result by showing that \bdvd is~$\W[1]$-hard parameterized by the treedepth.
The reduction even gives the stronger result that \bdvd is $\W[1]$-hard for the combined parameter treedepth + distance to cograph + distance to interval graph.
The \emph{treedepth} of a graph $G$ is the minimum height of a rooted forest $F$ such that
$G \subseteq \clos(F)$, where $\clos(F)$ denotes the graph with vertex set $V(F)$ and edge set~$
 \bigl\{\{u,v\} \mid \text{$u$ is an ancestor of $v$ in $F$ and $u \neq v$}\bigr\}$~\cite{JM14}.
A graph~$G$ is a \emph{cograph} if it does not contain an induced~$P_4$~\cite{BLS99}.
A graph is called chordal if every induced subgraph that is a cycle has exactly three vertices.
A graph $G$ is called an \emph{interval graph} if it is chordal and, for every triple of distinct vertices $v_1, v_2, v_3 \in V(G)$ such that no two of them are adjacent, at least one vertex $v_i$ is adjacent to every path connecting the other two vertices~\cite{LEK62}.

To establish $\W[1]$-hardness for the parameters mentioned above we provide a reduction from \rubp{}, which is defined as follows \cite{HANA24}.

\problemDef{\rubp}
{
	A set $S = \{s_1, \dots, s_n\}$ of integers in unary encoding, $k \in \ZZ^+$, as well as a
	function~$f : S \to \binom{[k]}2 \ceq \{\{x,y\} \mid x,y \in \{1,\dots,k\}, x \neq y\}$, that maps every $s_i \in S$ to a set of exactly
	two integers from $[k]$.
}
{
	Determine whether we can partition $S$ into $k$ subsets $S_1, \dots, S_k$
	such that for all
	$i \in [k]$ it holds that
	\begin{itemize}
		\item[(i)] $\sum_{s \in S_i} s = \frac 1 k \sum_{s \in S} s$
		\item[(ii)] For all $s \in S_i$, $i \in f(s)$.
	\end{itemize}
}

\begin{theorem}
	\label{thm:whard_td_co_interval}
    \BDVD{} is $\W[1]$-hard with respect to the combined parameter treedepth + distance to cograph + distance to interval graph.
\end{theorem}

\begin{proof}
	We provide a reduction from \rubp which is known to be $\W[1]$-hard parameterized
	by the number of bins $k$ \cite{HANA24}.
	Let $(S, k, f)$ be an instance of \rubp{}.
	We construct an instance $(G, k', \taup)$ of \bdvd{} as follows.

\begin{enumerate}
	\item For every $p \in [k]$, create a corresponding vertex $v_p$ in $G$.
	\item For each $s_i \in S$ with $f(s_i) = \{p, q\}$, introduce two vertices $x^i_{p}$
		and~$x^i_{q}$ in~$G$, making~$x^i_p$ adjacent to~$v_q$ and~$x^i_q$ adjacent to~$v_p$.
		We also make~$x_p^i$ adjacent to~$x_q^i$.
		Then add a set~$V_i = \{y^i_1, \dots, y^i_{3s_i-1}\}$
		of $3s_i - 1$ vertices to~$G$, and make each of~$x^i_p$ and~$x^i_q$ adjacent
		to every vertex in~$V_i$. We call the induced subgraph
		$G[\{x^i_p\} \cup \{x^i_q\} \cup V_i]$ a
		choice gadget for $s_i$. An example can be seen in Figure \ref{fig:choicegadget}.
	\item Set $k' = |S|$ and $\taup = 1 - \nicefrac{1}{3r+1}$, where
		$r = \sum_{s_i \in S} s_i / k$.
\end{enumerate}

\begin{figure*}[t]
\centering
	\begin{tikzpicture}[scale=1,dot/.style={draw,circle,inner sep=6pt}]
		\node[dot,label={$v_p$}] at (0,0) (u) {};
		\node[dot,label={$v_q$}] at (6, 0) (v) {};

		\node[dot,label={$x^i_q$}] at (1, 0) (xu) {};
		\node[dot,label={$x^i_p$}] at (5, 0) (xv) {};
		\node[dot] at (3, -0.5) (y1) {};
		\node[dot] at (3, 0.5) (y2) {};
		\node at (3.7,1.2) (label) {$V_i$};

		\draw (xu) -- (u);
		\draw (xv) -- (v);
		\draw (xv) edge [bend left=90] (xu);

		\draw 	(xu) -- (y1)
			(xu) -- (y2);

		\draw	(xv) -- (y1)
			(xv) -- (y2);

		\draw [rounded corners, dotted, line width=0.3mm] (2.5, 1) rectangle (3.5, -1);
	\end{tikzpicture}
	\caption{Choice gadget for item~$s_1=1$ with~$f(s_i) = \{p,q\}$ together with vertices~$v_p$ and~$v_q$ representing bins~$p$ and~$q$ in the reduction from \cref{thm:whard_td_co_interval}.
	\label{fig:choicegadget}
	}
\end{figure*}

Clearly, $(G, k', \taup)$ can be constructed in polynomial time in the size of $(S, k, f)$ when
all integers in $S$ are encoded in unary. Moreover, we show that all parameters mentioned in the theorem are bounded by a function of $k$.

For the treedepth we construct a rooted forest $F$ as follows. First, we add the path $(v_1,\dots,v_k)$ to~$F$ and
root~$F$ at~$v_1$.
Then, for each $s_i \in S$ with $f(s_i)=\{p,q\}$, we connect~$x_p^i$ to~$x_q^i$ and make~$x_q^i$
adjacent to every vertex $y^i \in V_i$.
Finally, we connect~$x_p^i$ to~$v_k$.
The resulting rooted forest has height~$k+3$ and satisfies $G \subseteq \clos(F)$.
Hence,~$G$ has treedepth at most~$k+3$.

By construction it is clear that no choice gadget contains an induced~$P_4$ and therefore is a cograph.
Moreover, any induced cycle in a choice gadget consists of exactly three
vertices, namely one vertex from $V_i$ together with $x_u^i$ and $x_v^i$.
Hence, each choice gadget is chordal.
Since both $x_u^i$ and $x_v^i$ are adjacent to all other vertices of the gadget, and every edge is incident to either $x_u^i$ or $x_v^i$, it follows that each choice gadget is an interval graph.
Since the removal of the~$k$ vertices corresponding to the bins leaves a graph where each component is a choice gadget, it follows that the distance to cograph/interval graph is at most~$k$.

\smallskip

Next, we prove the correctness of the reduction, that is, we show that $(S, k, f)$
is a yes-instance of \rubp{} if and only if $(G, k', \taup)$ is a yes-instance of \BDVD{}.

\medskip

``$\Rightarrow$'': Suppose there exists a solution $S_1, \dots, S_k$ of the \rubp-instance~$(S,k,f)$. We define a solution for the \bdvd-instance by
\begin{align*}
	V' = \bigcup_{p \in [k]} \bigcup_{s_i \in S_p} \{x^i_p\}.
\end{align*}

Clearly, $|V'| = |S| = k'$ as~$V'$ contains one vertex for each item in~$S$.
It suffices to show that $G - V'$ contains no subgraph with density exceeding
$\taup$. Since $G - V'$ is a forest, it is enough to verify that every connected component of
$G - V'$ has density at most $\taup$.

Observe that each component of $G - V'$ contains exactly one vertex from
$\{v_1, \dots, v_k\}$, each corresponding to a bin. Let $v_p$ be such a vertex.
For any $s_i \in S$ with $f(s_i) = \{p, q\}$ and $s_i \in S_q$ with $p \ne q$,
the vertex $x^i_q$ adjacent to $v_p$ has been deleted.
Hence, we only need to consider the $s_i \in S_p$.
Let $G_p \subseteq G - V'$ be the connected component containing~$v_p$. Its vertex set is~$V(G_p) = \{v_p\} \cup \bigcup_{s_i \in S_p} (\{x^i_q\} \cup V_i)$ and therefore~$|V(G_p)| = 1 + \sum_{s_i \in S_p} 3s_i = 1 + 3r.$
Since $G_p$ is a tree, we have $|E(G_p)| = |V(G_p)| - 1 = 3r$, and thus
\[
	\rho(G_p) = \frac{|E(G_p)|}{|V(G_p)|}
	= \frac{3r}{3r+1}
	= \taup.
\]
This completes the forward direction.

\medskip

``$\Leftarrow$'': Let~$V'$ be a solution to $(G, k', \taup)$. We first claim that~$V'$ must contain exactly one vertex from each choice gadget in~$G$.

Indeed, each choice gadget contains a cycle, so its density is at least~$1$. Since $\taup < 1$,
at least one vertex from every choice gadget must be deleted. There are~$k'$ pairwise disjoint
choice gadgets and therefore exactly one vertex is deleted from each gadget, and no vertices are deleted
from the remaining graph.

Next, we claim that the deleted vertex in any gadget must be either~$x^i_p$ or~$x^i_q$.
Since $s_i \ge 1$, we have $|V_i| \ge 2$. Moreover, any vertex in~$V_i$ forms a
cycle together with~$x^i_p$ and~$x^i_q$. Hence, deleting a single vertex from~$V_i$ does
not destroy all cycles in the gadget, so the deleted vertex must be one of~$x^i_p$ or~$x^i_q$.

This allows us to define a partition $S_1, \dots, S_k$: for each~$s_i$, we add~$s_i$ to~$S_p$ if~$x^i_p$ is contained in~$V'$.

Now observe that~$G-V'$ has exactly~$k$ components and since~$V'$ is a solution we know that each component has size at most~$3r+1$.
By construction~$G-V'$ has exactly~$k(3r+1)$ vertices and subsequently we know that each component in~$G-V'$ has size exactly~$3r+1$.

Let~$G_p$ be such a component containing the vertex~$v_p$.
Then~$|V(G_p)| = 3r+1$ which implies~$\sum_{s \in S_p} (|V_s| + 1) = 3r$ and therefore~$\sum_{s \in S_p} s = r$.

Hence, $S_1, \dots, S_k$ satisfies requirement (i), and (ii) is enforced by construction, completing the backward direction.
\end{proof}

\subsection{Treedepth and Feedback Vertex Number}
The feedback vertex number of a graph~$G$ is the minimum cardinality of any set~$X \subseteq V(G)$ such that~$G-X$ is acyclic~\cite{PARAMALG}.
It is a smaller parameter than the vertex cover number and the max leaf number, for which \bdvd is in~$\FPT$~(\cite{BNV25} and \cref{thm:max_leaf_number}).
Note that in the proof of \cref{thm:whard_td_co_interval} we actually constructed a feedback vertex set, but of size not bounded by the parameter.
Here, we establish that \bdvd is $\W[1]$-hard for the combined parameter feedback vertex number + treedepth.

\begin{theorem}
	\label{thm:whard_td_fvn}
    	\BDVD{} is $\W[1]$-hard with respect to the combined parameter
	treedepth + feedback vertex number.
\end{theorem}

{

\begin{proof}
	We construct a similar reduction to Hanaka et al.~\cite{HANA24} and reduce from
	\bdd{} (known to be $\W[1]$-hard parameterized by treedepth + feedback vertex number \cite{BDD}) to
	\BDVD{}.

\problemDef{\bdd{}}
{An undirected Graph $G$, a degree bound $d$ and a solution size $k$.}
{
	Is there a subset $S \subseteq V(G)$ with $|S| \le k$ such that $G-S$
	contains no vertex with a degree of more than $d$?
}

	Let $(G, k, d)$ be an instance of \bdd{}. We construct an instance $(G', k', \taup)$
	of \BDVD{} as follows.
	\begin{itemize}
		\item[i.] Add a copy of each vertex $v \in V(G)$ to $G'$.
		\item[ii.] For every edge $\{u, v\} \in E(G)$ we add three vertices
			$u_{uv}, v_{uv}$ and~$y_{uv}$ to $G'$ and connect~$y_{uv}$ to~$u_{uv}$
			and~$v_{uv}$. We also connect~$u$ to~$u_{uv}$ and~$v$ to~$v_{uv}$.

			Additionally we add a set~$V_{uv}$ of $d-1$ vertices to~$G'$ and
			connect each vertex to~$y_{uv}$. An example can be seen in
			\cref{fig:edge-gadget}.
		\item[iii.] At last we set $k' = k + |E(G)|$ and $\taup = \nicefrac{d}{d+1}$.
	\end{itemize}
	\begin{figure*}[t]
	\centering
	\begin{tikzpicture}[scale=1,dot/.style={draw,circle,inner sep=6pt}]
		\node[dot,label={$v$}] at (-1,0) (u) {};
		\node[dot,label={$u$}] at (7, 0) (v) {};

		\node[dot,label={$v_{uv}$}] at (0, 0) (uuv) {};
		\node[dot,label={$u_{uv}$}] at (6, 0) (vuv) {};
		\node[dot,label={$y_{uv}$}] at (3, 0) (yuv) {};
		\node[dot] at (1, -1) (v1) {};
		\node[dot] at (2, -1) (v2) {};
		\node[dot] at (3, -1) (v3) {};
		\node[dot] at (4, -1) (v4) {};
		\node[dot] at (5, -1) (v5) {};
		\node at (6,-1) (label) {$V_{uv}$};

		\draw (uuv) -- (u);
		\draw (vuv) -- (v);
		\draw (uuv) -- (yuv);
		\draw (vuv) -- (yuv);

		\draw
			(yuv) -- (v1)
			(yuv) -- (v2)
			(yuv) -- (v3)
			(yuv) -- (v4)
			(yuv) -- (v5);

		\draw [rounded corners, dotted, line width=0.3mm] (5.5, -0.5) rectangle (0.5, -1.5);
	\end{tikzpicture}
	\caption{The corresponding gadget from $G'$ for a single edge $\{u, v\}$ in $G$ (and $d = 6$) from the proof of~\cref{thm:whard_td_fvn}.
	\label{fig:edge-gadget}
	}
	\end{figure*}
	Hanaka et al.~\cite{HANA24} use the same construction of~$G'$, reducing \bdd{} to \coc{}.
	As they show, for every edge gadget, at least one of the vertices
	$u_{uv}, v_{uv}$, or $y_{uv}$ must be included in any solution~$S$, as
	otherwise the induced subgraph~$G[\{u_{uv}, v_{uv}, y_{uv}\} \cup V_{uv}]$
	would be too large with respect to the \coc{} constraint.
	Consequently, for any feasible solution~$S$, the resulting graph $G' - S$
	must necessarily be a forest.

	Since \BDVD{} with $\taup < 1$ merely restricts the component size
	of the remaining forest, the correctness of the reduction follows directly from the arguments
	of Hanaka et al.~\cite{HANA24}.
\end{proof}
}

\section{Conclusion}
We filled several gaps in the parameterized complexity landscape of \BDVD, including the resolution of an open question of Bazgan et al.~\cite{BNV25} about the fixed-parameter tractability with respect to treewidth, see \Cref{fig:results} for an overview.
A common theme in our results is that restricting the target density~$\taup$ makes a significant difference in the complexity of \bdvd.
Without restricting the target density, when considering parameters related to (very) sparse graph such as feedback vertex number, feedback edge number, or max leaf number, the challenging case seems to be~$\taup < 1$.
Then \bdvd behaves like a balanced partitioning problem, which is hard even in very sparse graphs~\cite{BNV25,EFGK2009}. 

We leave open the question of the parameterized complexity of \bdvd for parameters between the ($\W[1]$-hard) distance to cograph and distance to interval graph and the (fixed-parameter tractable) vertex cover number, e.g. distance to cluster and twin-cover number.
Another open question posed by Bazgan et al.~\cite{BNV25} remains whether \bdvd{} is hard for the combined
parameter treewidth and solution size.
Since the sister problem \bded{} is known to be $\W[1]$-hard when parameterized by treewidth and solution size, it seems plausible that \bdvd{} exhibits similar hardness.
Similarly the parameterized complexity of \bdvd for the feedback edge number (for which \bded is $\W[1]$-hard) is an open question.

\newpage
\bibliographystyle{plainurl} %
\bibliography{References.bib}

@article{BNV25,
	title = {Destroying densest subgraphs is hard},
	journal = {Journal of Computer and System Sciences},
	volume = {151},
	year = {2025},
	issn_IGNORE = {0022-0000},
	url = {https://doi.org/10.1016/j.jcss.2025.103635},
	author = {Cristina Bazgan and André Nichterlein and Sofia {Vazquez Alferez}},
}

@inproceedings{CHAR00,
  author       = {Moses Charikar},
  title        = {Greedy approximation algorithms for finding dense components in a
                  graph},
  booktitle    = {Proceedings of the Third International Workshop on Approximation Algorithms for Combinatorial Optimization ({APPROX} 2000)},
  address_IGNORE = {Saarbrücken, Germany},
  series       = {LNCS},
  volume_IGNORE       = {1913},
  publisher    = {Springer},
  year         = {2000},
  url          = {https://doi.org/10.1007/3-540-44436-X_10},
  doi_IGNORE          = {10.1007/3-540-44436-X_10},
}

@book{PARAMALG,
	author = {Cygan, Marek and Fomin, Fedor V. and Kowalik, Lukasz and Lokshtanov, Daniel and Marx, Daniel and Pilipczuk, Marcin and Pilipczuk, Michal and Saurabh, Saket},
	title = {Parameterized Algorithms},
	year = {2015},
	isbn = {3319212745},
	publisher = {Springer},
	edition = {1st},
}

@book{DF13,
  author       = {Rodney G. Downey and
                  Michael R. Fellows},
  title        = {Fundamentals of Parameterized Complexity},
  publisher    = {Springer},
  year         = {2013},
  url          = {https://doi.org/10.1007/978-1-4471-5559-1},
  ISBN         = {978-1-4471-7164-5}
}

@InProceedings{HANA24,
  author =	{Hanaka, Tesshu and Lampis, Michael and Vasilakis, Manolis and Yoshiwatari, Kanae},
  title =	{{Parameterized Vertex Integrity Revisited}},
  booktitle =	{Proceedings of the 49th International Symposium on Mathematical Foundations of Computer Science (MFCS 2024)},
  series =	{Leibniz International Proceedings in Informatics (LIPIcs)},
  ISBN_IGNORE =	{978-3-95977-335-5},
  ISSN_IGNORE =	{1868-8969},
  year =	{2024},
  volume_IGNORE =	{306},
  publisher =	{Schloss Dagstuhl -- Leibniz-Zentrum für Informatik},
  address_IGNORE =	{Dagstuhl, Germany},
  URL =		{https://doi.org/10.4230/LIPIcs.MFCS.2024.58 },
  URN =		{urn:nbn:de:0030-drops-206141},
  doi_IGNORE =		{10.4230/LIPIcs.MFCS.2024.58}
}

@article{LANC24,
	author = {Lanciano, Tommaso and Miyauchi, Atsushi and Fazzone, Adriano and Bonchi, Francesco},
	title = {A Survey on the Densest Subgraph Problem and its Variants},
	year = {2024},
	issue_date = {August 2024},
	publisher = {Association for Computing Machinery},
	address_IGNORE = {New York, NY, USA},
	volume = {56},
	number = {8},
	issn_IGNORE = {0360-0300},
	url = {https://doi.org/10.1145/3653298},
	doi_IGNORE = {10.1145/3653298},
	journal = {ACM Computing Surveys},
	month = apr,
	articleno = {208}
}

@InProceedings{CHAN25,
  author =	{Chandrasekaran, Karthekeyan and Chekuri, Chandra and Kulkarni, Shubhang},
  title =	{{On Deleting Vertices to Reduce Density in Graphs and Supermodular Functions}},
  booktitle =	{Proceedings of the 52nd International Colloquium on Automata, Languages, and Programming (ICALP 2025)},
  series =	{Leibniz International Proceedings in Informatics (LIPIcs)},
  ISBN_IGNORE =	{978-3-95977-372-0},
  ISSN_IGNORE =	{1868-8969},
  year =	{2025},
  volume_IGNORE =	{334},
  publisher =	{Schloss Dagstuhl -- Leibniz-Zentrum für Informatik},
  address_IGNORE =	{Dagstuhl, Germany},
  URL =		{https://drops.dagstuhl.de/entities/document/10.4230/LIPIcs.ICALP.2025.43},
  URN =		{urn:nbn:de:0030-drops-234200},
  doi_IGNORE =		{10.4230/LIPIcs.ICALP.2025.43}
}

@book{JM14,
  title={Sparsity: graphs, structures, and algorithms},
  author={Nesetril, Jaroslav and De Mendez, Patrice Ossona},
  year={2014},
  publisher={Springer},
  url={https://doi.org/10.1007/978-3-642-27875-4},
}

@article{LEK62,
	author = {Lekkeikerker, Cornelis and Boland, Johan},
	journal = {Fundamenta Mathematicae},
	number = {1},
	title = {Representation of a finite graph by a set of intervals on the real line},
	url = {http://eudml.org/doc/213681},
	volume = {51},
	year = {1962}
}

@inproceedings{NICH26,
  author       = {Matthias Bentert and
                  Tom{-}Lukas Breitkopf and
                  Vincent Froese and
                  Anton Herrmann and
                  Andr{\'{e}} Nichterlein},
  title        = {Density Matters: {A} Complexity Dichotomy of Deleting Edges to Bound
                  Subgraph Density},
  booktitle    = {Proceedings of the 43rd International Symposium on Theoretical Aspects of Computer Science ({STACS} 2026)},
  address_IGNORE      = {Grenoble, France},
  series       = {LIPIcs},
  volume_IGNORE       = {364},
  publisher    = {Schloss Dagstuhl - Leibniz-Zentrum für Informatik},
  year         = {2026},
  url          = {https://doi.org/10.4230/LIPIcs.STACS.2026.12},
}

@article{BDD,
	author = {Ganian, Robert and Klute, Fabian and Ordyniak, Sebastian},
	journal = {Algorithmica},
	number = {1},
	title = {On Structural Parameterizations of the Bounded-Degree Vertex Deletion Problem},
	volume = {83},
	year = {2021},
	url = {https://doi.org/10.1007/s00453-020-00758-8},
	doi_IGNORE = {10.1007/s00453-020-00758-8},
}

@article{AsahiroJMO15,
  author       = {Yuichi Asahiro and
                  Jesper Jansson and
                  Eiji Miyano and
                  Hirotaka Ono},
  title        = {Graph Orientations Optimizing the Number of Light or Heavy Vertices},
  journal      = {Journal of Graph Algorithms and Applications},
  volume       = {19},
  number       = {1},
  pages        = {441--465},
  year         = {2015},
  url          = {https://doi.org/10.7155/jgaa.00371},
  doi_IGNORE          = {10.7155/JGAA.00371}
}

@article{BodlaenderOO18,
  author       = {Hans L. Bodlaender and
                  Hirotaka Ono and
                  Yota Otachi},
  title        = {Degree-Constrained Orientation of Maximum Satisfaction: Graph Classes
                  and Parameterized Complexity},
  journal      = {Algorithmica},
  volume       = {80},
  number       = {7},
  pages        = {2160--2180},
  year         = {2018},
  url          = {https://doi.org/10.1007/s00453-017-0399-9},
  doi          = {10.1007/S00453-017-0399-9}
}

@article{HKLOS20,
  author       = {Tesshu Hanaka and Ioannis Katsikarelis and Michael Lampis and Yota Otachi and Florian Sikora},
  title        = {Parameterized Orientable Deletion},
  journal      = {Algorithmica},
  volume       = {82},
  number       = {7},
  pages        = {1909--1938},
  year         = {2020},
  url          = {https://doi.org/10.1007/s00453-020-00679-6},
}

@article{PQ82,
  author =        {Jean{-}Claude Picard and Maurice Queyranne},
  journal =       {Networks},
  number =        {2},
  pages =         {141--159},
  title =         {A network flow solution to some nonlinear 0-1 programming problems, with applications to graph theory},
  volume =        {12},
  year =          {1982},
  url =           {https://doi.org/10.1002/NET.3230120206},
}

@techreport{Gol84,
  address_IGNORE =       {USA},
  author =        {Goldberg, Andrew V.},
  institution =   {University of California at Berkeley},
  title =         {Finding a Maximum Density Subgraph},
  year =          {1984},
}

@article{AJMO16,
  author       = {Yuichi Asahiro and Jesper Jansson and Eiji Miyano and Hirotaka Ono},
  title        = {Degree-Constrained Graph Orientation: Maximum Satisfaction and Minimum Violation},
  journal      = {Theory of Computing Systems},
  volume       = {58},
  number       = {1},
  pages        = {60--93},
  year         = {2016},
  url          = {https://doi.org/10.1007/s00224-014-9565-5},
  doi_IGNORE          = {10.1007/S00224-014-9565-5},
}

@inproceedings{Korhonen024,
  author       = {Tuukka Korhonen and Marek Sokolowski},
  editor       = {Bojan Mohar and Igor Shinkar and
                  Ryan O'Donnell},
  title        = {Almost-Linear Time Parameterized Algorithm for Rankwidth via Dynamic
                  Rankwidth},
  booktitle    = {Proceedings of the 56th Annual {ACM} Symposium on Theory of Computing,
                  ({STOC} 2024)},
  address_IGNORE_IGNORE = {Vancouver, Canada},
  pages        = {1538--1549},
  publisher    = {{ACM}},
  year         = {2024},
  url          = {https://doi.org/10.1145/3618260.3649732},
  doi_IGNORE          = {10.1145/3618260.3649732}
}

@article{CourcelleMR00,
  author       = {Bruno Courcelle and
                  Johann A. Makowsky and
                  Udi Rotics},
  title        = {Linear Time Solvable Optimization Problems on Graphs of Bounded Clique-Width},
  journal      = {Theory of Computing Systems},
  volume       = {33},
  number       = {2},
  pages        = {125--150},
  year         = {2000},
  url          = {https://doi.org/10.1007/s002249910009},
  doi_IGNORE          = {10.1007/S002249910009}
}

@article{GimaHKKO22,
  author       = {Tatsuya Gima and
                  Tesshu Hanaka and
                  Masashi Kiyomi and
                  Yasuaki Kobayashi and
                  Yota Otachi},
  title        = {Exploring the gap between treedepth and vertex cover through vertex
                  integrity},
  journal      = {Theoretical Computer Science},
  volume       = {918},
  pages        = {60--76},
  year         = {2022},
  url          = {https://doi.org/10.1016/j.tcs.2022.03.021},
  doi_IGNORE          = {10.1016/J.TCS.2022.03.021}
}

@article{DrangeDH16,
  author       = {P{\aa}l Gr{\o}n{\aa}s Drange and
                  Markus S. Dregi and
                  Pim van 't Hof},
  title        = {On the Computational Complexity of Vertex Integrity and Component
                  Order Connectivity},
  journal      = {Algorithmica},
  volume       = {76},
  number       = {4},
  pages        = {1181--1202},
  year         = {2016},
  url          = {https://doi.org/10.1007/s00453-016-0127-x},
  doi_IGNORE          = {10.1007/S00453-016-0127-X}
}

@article{FrankT87,
  author       = {Andr{\'{a}}s Frank and
                  {\'{E}}va Tardos},
  title        = {An application of simultaneous Diophantine approximation in combinatorial
                  optimization},
  journal      = {Comb.},
  volume       = {7},
  number       = {1},
  pages        = {49--65},
  year         = {1987},
  url          = {https://doi.org/10.1007/BF02579200},
  doi_IGNORE          = {10.1007/BF02579200}
}

@article{Kannan87,
  author       = {Ravi Kannan},
  title        = {Minkowski's Convex Body Theorem and Integer Programming},
  journal      = {Mathematics of Operations Research},
  volume       = {12},
  number       = {3},
  pages        = {415--440},
  year         = {1987},
  url          = {https://doi.org/10.1287/moor.12.3.415},
  doi_IGNORE          = {10.1287/MOOR.12.3.415}
}

@article{Lenstra83,
  author       = {Hendrik W. Lenstra Jr.},
  title        = {Integer Programming with a Fixed Number of Variables},
  journal      = {Mathematics of Operations Research},
  volume       = {8},
  number       = {4},
  pages        = {538--548},
  year         = {1983},
  url          = {https://doi.org/10.1287/moor.8.4.538},
  doi_IGNORE          = {10.1287/MOOR.8.4.538}
}

@Book{BLS99,
  author        = {Andreas Brandst{\"a}dt and Van Bang Le and Jeremy P. Spinrad},
  title         = {Graph Classes: a Survey},
  publisher     = {{SIAM}},
  series        = {{SIAM Monographs on Discrete Mathematics and Applications}},
  volume        = {3},
  year          = {1999}
}

@inproceedings{GarvardtK24,
  author       = {Jaroslav Garvardt and
                  Christian Komusiewicz},
  editor       = {{\'{E}}douard Bonnet and
                  Pawel Rzazewski},
  title        = {Modularity Clustering Parameterized by Max Leaf Number},
  booktitle    = {Proceedings of the 19th International Symposium on Parameterized and Exact Computation ({IPEC} 2024)},
  address_IGNORE = {Egham, United Kingdom},
  series       = {LIPIcs},
  pages        = {16:1--16:14},
  publisher    = {Schloss Dagstuhl - Leibniz-Zentrum f{\"{u}}r Informatik},
  year         = {2024},
  url          = {https://doi.org/10.4230/LIPIcs.IPEC.2024.16},
  doi_IGNORE          = {10.4230/LIPICS.IPEC.2024.16}
}

@article{Eppstein15,
  author       = {David Eppstein},
  title        = {Metric Dimension Parameterized by Max Leaf Number},
  journal      = {Journal of Graph Algorithms and Applications},
  volume       = {19},
  number       = {1},
  pages        = {313--323},
  year         = {2015},
  url          = {https://doi.org/10.7155/jgaa.00360},
  doi_IGNORE          = {10.7155/JGAA.00360}
}

@article{CR05,
  author       = {Derek G. Corneil and
                  Udi Rotics},
  title        = {On the Relationship Between Clique-Width and Treewidth},
  journal      = {{SIAM} Journal on Computing},
  volume       = {34},
  number       = {4},
  pages        = {825--847},
  year         = {2005},
  url          = {https://doi.org/10.1137/S0097539701385351},
  doi_IGNORE          = {10.1137/S0097539701385351},
}

@article{CO00,
  author       = {Bruno Courcelle and
                  Stephan Olariu},
  title        = {Upper bounds to the clique width of graphs},
  journal      = {Discret. Appl. Math.},
  volume       = {101},
  number       = {1-3},
  pages        = {77--114},
  year         = {2000},
  url          = {https://doi.org/10.1016/S0166-218X(99)00184-5},
  doi_IGNORE          = {10.1016/S0166-218X(99)00184-5},
}

@inproceedings{EFGK2009,
  author       = {Rosa Enciso and
                  Michael R. Fellows and
                  Jiong Guo and
                  Iyad A. Kanj and
                  Frances A. Rosamond and
                  Ondrej Such{\'{y}}},
  title        = {What Makes Equitable Connected Partition Easy},
  booktitle    = {In Proceedings of the 4th International Workshop on Parameterized and Exact Computation ({IWPEC} 2009)},
  series       = {LNCS},
  volume_IGNORE       = {5917},
  pages        = {122--133},
  publisher    = {Springer},
  year         = {2009},
  address_IGNORE      = {Copenhagen, Denmark}
}

\newpage

\appendix

\section{Missing Proof Details}
\appendixProofs

\end{document}